\documentclass[journal]{IEEEtran}
\usepackage{mathpazo}
\usepackage{times}

\usepackage{amsmath}
\usepackage{amsfonts}
\usepackage{latexsym}
\usepackage{amssymb}
\usepackage{upref}
\usepackage{theorem}
\usepackage{graphicx}
\usepackage{psfrag}
\usepackage{cite}
\usepackage{epsfig}
\usepackage{color}
\usepackage{graphicx}
\usepackage{colortbl}

\makeindex

\newcommand{\tab}{\hspace*{1em}}
\theoremstyle{plain}
\theorembodyfont{\normalfont\slshape}

\newtheorem{thm}{Theorem$\!$}
\newenvironment{theorem}
{\begin{thm}\hspace*{-1ex}{\bf.}}{\end{thm}}

\newtheorem{clm}[thm]{Claim$\!$}

\newtheorem{lem}[thm]{Lemma$\!$}

\newtheorem{prop}[thm]{Proposition$\!$}

\newtheorem{cor}[thm]{Corollary$\!$}

\newtheorem{defn}[thm]{Definition$\!$}

\newtheorem{xmpl}{Example$\!$}

\newtheorem{cnstr}{Construction$\!$}

\newtheorem{alg}{Algorithm$\!$}

\newcounter{enumrom}
\renewcommand{\theenumrom}{(\roman{enumrom})}

\makeatletter
\renewcommand{\@endtheorem}{\endtrivlist}
\makeatother

\makeatletter
\renewcommand{\thefigure}{{\@arabic\c@figure}}
\renewcommand{\fnum@figure}{{\bf Figure\,\thefigure}}
\makeatother
\newcommand{\cC}{{\cal C}}

\newcommand{\cM}{{\cal M}}

\newcommand{\cW}{{\cal W}}

\DeclareMathOperator{\sgn}{sgn}
\DeclareMathOperator{\spun}{span}

\begin{document}

%\begin{frontmatter}

% paper title
\title{\textbf{Zigzag Codes: MDS Array Codes with Optimal Rebuilding}
\vspace*{-0.2ex}}

\author{\IEEEauthorblockN{Itzhak Tamo\IEEEauthorrefmark{1}\IEEEauthorrefmark{2}, Zhiying Wang\IEEEauthorrefmark{1}, and Jehoshua Bruck\IEEEauthorrefmark{1}\\}
\IEEEauthorblockA{\IEEEauthorrefmark{1}Electrical Engineering Department,
California Institute of Technology,
Pasadena, CA 91125, USA \\}
\IEEEauthorblockA{\IEEEauthorrefmark{2}Electrical and Computer Engineering,
Ben-Gurion University of the Negev,
Beer Sheva 84105, Israel\\}
{\it \{tamo, zhiying, bruck\}@caltech.edu}\vspace*{-2.0ex}}

\maketitle

\begin{abstract}
MDS array codes are widely used in storage systems to protect data against erasures. We address the \emph{rebuilding ratio} problem, namely, in the case of erasures, what is the fraction of the remaining information that needs to be accessed in order to rebuild \emph{exactly} the lost information? It is clear that when the number of erasures equals the maximum number of erasures that an MDS code can correct then the rebuilding ratio is $1$ (access all the remaining information). However, the interesting and more practical case is when the number of erasures is smaller than the erasure correcting capability of the code. For  example, consider an MDS code that can correct two erasures: What is the smallest amount of information that one needs to access in order to correct a single erasure? Previous work showed that the rebuilding ratio is bounded between $\frac{1}{2}$ and $\frac{3}{4}$, however, the exact value was left as an open problem. In this paper, we solve this open problem and prove that for the case of a single erasure with a $2$-erasure correcting code, the rebuilding ratio is $\frac{1}{2}$. In general, we construct a new family of $r$-erasure correcting MDS array codes that has optimal rebuilding ratio of $\frac{e}{r}$ in the case of $e$ erasures, $1 \le e \le r$. Our array codes have efficient encoding and decoding algorithms (for the case $r=2$ they use a finite field of size $3$) and an optimal update property.
\end{abstract}

%\begin{keyword}\end{keyword}

%\end{frontmatter}

\section{Introduction}
Erasure-correcting codes are the basis of the ubiquitous RAID schemes for storage systems, where disks correspond to symbols in the code. Specifically, RAID schemes are based on MDS (maximum distance separable) array codes that enable optimal storage and efficient encoding and decoding algorithms. With $r$ redundancy symbols, an MDS code is able to reconstruct the original information if no more than $r$ symbols are erased. An array code is a two dimensional array, where each column corresponds to a symbol in the code and is stored in a disk in the RAID scheme. We are going to refer to a disk/symbol as a node or a column interchangeably, and an entry in the array as an element. Examples of MDS array codes are EVENODD \cite{Shuki-evenodd,Blaum96mdsarray}, B-code \cite{B-code}, X-code \cite{ x-code}, RDP \cite{RDP-code}, and STAR-code \cite{star-code}.

Suppose that some nodes are erased in a systematic MDS array code, we will rebuild them by accessing (reading) some information in the surviving nodes, all of which  are  assumed to be accessible. The fraction of the accessed information in the surviving nodes is called the \emph{rebuilding  ratio}. If $r$ nodes are erased, then the rebuilding ratio is $1$ since we need to read all the remaining information. Is it possible to lower this ratio for less than $r$ erasures? Apparently, it is possible: Figure \ref{fig:firstFigure} shows an
example of our new MDS code with $2$ information nodes and $2$ redundancy nodes, every node has $2$ elements, and operations are over a finite field of size $3$. Consider the rebuilding of the first information node, it requires access to $3$ elements out of $6$ (a rebuilding ratio of $\frac{1}{2}$), because $a= (a+b)-b$ and $c= (c+b)-b$. In practice, there is a difference between erasures of the information (also called systematic) and the parity nodes. An erasure of the former will affect the information access time since part of the raw information is missing, however erasure of the latter does not have such effect, since the entire information is accessible.  Moreover, in most storage systems the number of parity nodes is negligible compared to the number of systematic nodes. Therefore our constructions focus on the optimally of the rebuilding ratio related to the systematic nodes.

\begin{figure}[htp]
	\centering
	\includegraphics[width=.47\textwidth]{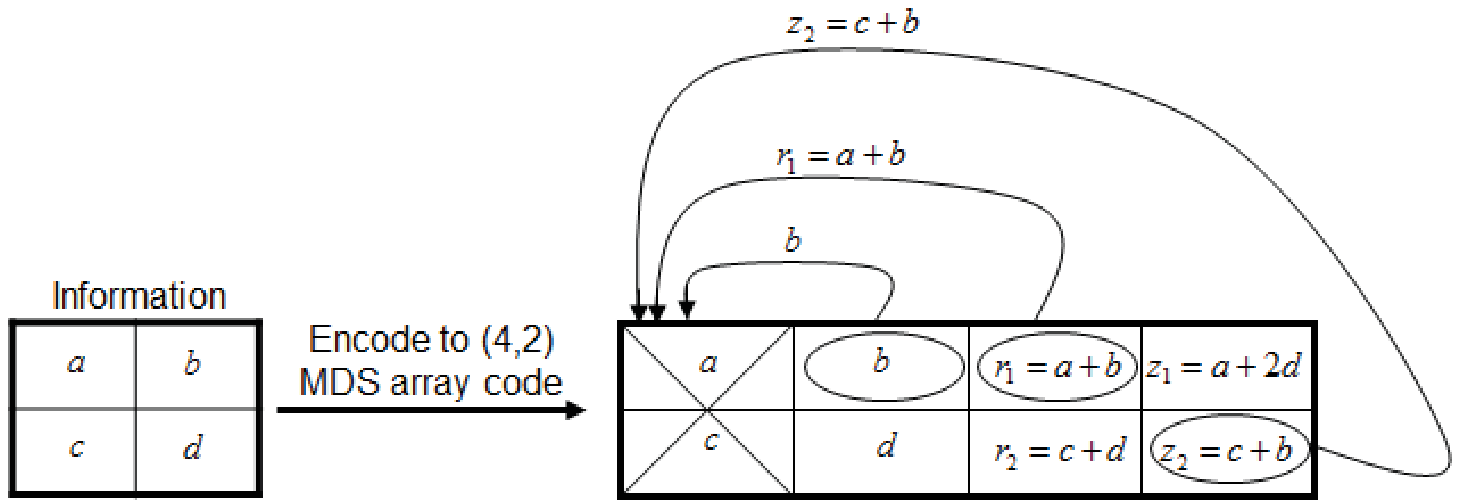}
\caption{Rebuilding of a $(4,2)$ MDS array code over $\mathbb{F}_{3}$. Assume the first node (column) is erased. }
	\label{fig:firstFigure}
	\vspace{-0.1cm}
	\end{figure}

In \cite{Dimakis2010,Wu07deterministicregenerating}, a related problem is discussed: The nodes are assumed to be distributed and fully connected in a network, and the concept of a \emph{repair bandwidth} is defined as the minimum amount of data that needs to be transmitted in the network in order to rebuild the erased nodes. In contrast to our concept of the \emph{rebuilding ratio} a transmitted element of data can be a function of a number of elements that are accessible on the same node. In addition, in their general framework, an acceptable rebuilding is one that retains the MDS property and not necessarily rebuilds the original erased node, whereas, we restrict our solutions to \emph{exact} rebuilding. It is clear that our framework is a special case of the general framework, hence, the repair bandwidth is a lower bound on the rebuilding ratio. What is known about lower bounds on the repair bandwidth? In \cite{Dimakis2010} it was proved that a lower bound on the repair bandwidth for an $(n,k)$ MDS code is:
\begin{equation} \label{eq:tradeoff}
\frac{\cM}{k}\cdot \frac{n-1}{n-k} .
\end{equation}
Here the code has a total of $n$ nodes with $k$ nodes of information and $r=n-k$ nodes of redundancy/parity, where $\cM$ is the total amount of information.
%An explanation on the notation in Equation (\ref{eq:tradeoff}): It considers an $(n,k)$ MDS code - with a total of $n$ nodes with $k$ nodes of information and $r=n-k$ nodes of redundancy (parity), where $\cM$ is the total amount of information,
Also all the surviving nodes are assumed to be accessible. It can be verified that Figure \ref{fig:firstFigure} matches this lower bound. Note that Equation (\ref{eq:tradeoff}) represents the amount of information, it should be normalized to reach the ratio. A number of researchers addressed the repair bandwidth problem \cite{Dimakis2010,5206008,Dimakis-interference-alignment,Kumar09,Suh-alignment,Kumar2009,Changho-Suh2010,Cadambe2010,Wu07deterministicregenerating,Rashmi11}, however the constructed code achieving the lower bound all have low code rate, i.e., $k/n<1/2$.
And it was shown by interference alignment in \cite{Cadambe2010,Changho-Suh2010} that this bound is asymptotically achievable for exact repair.

Instead of trying to construct MDS codes that can be easily rebuilt, a different approach \cite{zhiying10, XiangXLC10} was used by trying to find ways to rebuild existing families of MDS array codes. The ratio  of rebuilding a single systematic node was shown to be $\frac{3}{4}+o(1)$ for EVENODD or RDP codes\cite{Shuki-evenodd,RDP-code}, both of which have 2 parities. However, based on the lower bound of \eqref{eq:tradeoff} the ratio can be as small as $1/2$. Moreover, related work on constructing codes with optimal rebuilding appeared independently in \cite{viveck,Papailiopoulos}. Their constructions are similar to this work, but only single erasure is considered.

Our main goal in this paper is to design $(n,k)$ MDS array codes with \emph{optimal rebuilding ratio, for arbitrary number of parities.} We first consider the case of $2$ parities. We assume that the code is systematic. In addition, we consider codes with \emph{optimal update}, namely, when an information element is written, only the element itself and one element from each parity node needs update, namely, there is optimal reading/writing during writing of information. Hence, in the case of a code with $2$ parities only $3$ elements are updated. Under such assumptions, we will prove that every parity element is a linear combination of exactly one information element from each systematic column. We call this set of information elements a \emph{parity set}. Moreover, the parity sets of a parity node form a partition of the information array.

\begin{figure}
	\centering
\begin{tabular}{|c|c|c|c|c|c|}
	\hline
	& 0 & 1 & 2 & R & Z \\
	\hline
	0 & \cellcolor[gray]{0.8}{$\clubsuit$} & $\spadesuit$ & \cellcolor[gray]{0.8}{$\heartsuit$} & \cellcolor[gray]{0.8}{ } & \cellcolor[gray]{0.8}{$\clubsuit$} \\
	\hline
	1 & \cellcolor[gray]{0.8}{$\heartsuit$} & $\diamondsuit$ & \cellcolor[gray]{0.8}{$\clubsuit$} & \cellcolor[gray]{0.8}{ }& \cellcolor[gray]{0.8}{$\heartsuit$} \\
	\hline
	2 & $\spadesuit$ & $\clubsuit$ & $\diamondsuit$ & & $\spadesuit$ \\
	\hline
	3 & $\diamondsuit$ & $\heartsuit$ & $\spadesuit$ & & $\diamondsuit$ \\
	\hline
\end{tabular}
\caption{Permutations for zigzag sets in a $(5,3)$ code with $4$ rows. Columns 0, 1, and 2 are systematic nodes  and columns R, and Z are parity nodes. Each element in column R is a linear combination of the systematic elements in the same row. Each element in column Z is a linear combination of the systematic elements with the same symbol. The shaded elements are accessed to rebuild column 1.}
\label{fig:shapes}
%\vspace{-0.5cm}
\end{figure}

For example, Figure \ref{fig:shapes} shows a code with $3$ systematic nodes and $2$ parity nodes. The parity sets corresponding to parity node $R$ are the sets of information elements in the same row. The parity sets that correspond to the parity node $Z$ are the sets of information elements with the same symbol. For instance the first element in column $R$ is a linear combination of the elements in the first row and in columns $0,1$, and $2$.
And the $\clubsuit$ in column Z is a linear combination of all the $\clubsuit$ elements in columns $0, 1$, and $2$. We can see that each systematic column corresponds to a permutation of the four symbols. In general, we will show that each parity relates to a set of a permutations of the systematic columns. Without loss of generality,  we assume that the first parity node corresponds to identity permutations, namely, it is linear combination of rows. In the case of codes with $2$ parities, we call the first parity the \emph{row parity} and the second parity the \emph {zigzag parity}. The corresponding sets of information elements for a parity element are called \emph{row} and \emph{zigzag sets}, respectively.

It should be noted that in contrast to existing MDS array codes such as EVENODD and X-code, the parity sets in our codes are not limited to elements that correspond to straight lines in the array, but can also include elements that correspond to zigzag lines. We will demonstrate that this property is essential for achieving an optimal rebuilding ratio.

If a single systematic node is erased, we will rebuild each element in the erased node either by its corresponding row parity or zigzag parity, referred to as \emph{rebuild by row (or by zigzag)}. In particular, we access the row (zigzag) parity element, and all the elements in this row (zigzag) set, except the erased element. For example, consider Figure \ref{fig:shapes}, suppose that the column labeled $1$ is erased, one can access the $8$ shaded elements and rebuild its first two elements by rows, and the rest by zigzags. Namely, only half of the remaining elements are accessed. It can be verified that for the code in  Figure \ref{fig:shapes}, all the three systematic columns can be rebuilt by accessing half of the remaining elements. Thus the rebuilding ratio is $1/2$, which is the lower bound expressed in \eqref{eq:tradeoff}.

The key idea in our construction is that for each erased node, the row sets and the zigzag sets have a large intersection - resulting in
a small number of accesses. So the question is: How do we find permutations such that the row sets and zigzag sets intersect as much as possible? In this paper, we will present an optimal solution to this question by constructing permutations that are derived from binary vectors. This construction provides an optimal rebuilding ratio of $1/2$ for any erasure of a systematic node.
How do we define permutations on integers from a binary vector? We simply add to each integer the binary vector and use the sum as the image of this integer. Here each integer is expressed as its binary expansion. For example, in order to define the permutation on integers $\{0,1,2,3\}$ from the binary vector $v=(1,1)$, we express each integer in binary: $(0,0),(0,1),(1,0),(1,1)$. Then add (mod $2$) the vector $v=(1,1)$ to each integer, and get $(1,1),(1,0),(0,1),(0,0)$. At last change each binary expansion back to integer and define it as the image of the permutation: $3,2,1,0$. Hence, $(0,1,2,3)$ are mapped to $(3,2,1,0)$ in this permutation, respectively.
This simple technique for generating permutations is the key in our construction. We can generalize our construction for arbitrary $r$ (number of parity nodes) by generating permutations using $r$-ary vectors. Our constructions are optimal in the sense that  we can construct codes with $r$ parities and a rebuilding ratio of $1/r$.

So far we focused on the optimal rebuilding ratio, however, a code with two parity nodes should be able to correct two erasures, namely, it needs to be an MDS code. We will present that for large enough field size the code can be made MDS. In particular, another key result in this paper is that for the case of a code with two parity nodes, the field size is $3$, and this field size is optimal.

In addition, our codes have an optimal array size in the sense that for a given number of rows, we have the maximum number of columns among all systematic codes with optimal ratio and update. However, the length of the array is exponential in the width. We introduce techniques for making the array wider while having a rebuilding ratio that is very close to $1/r$.

 We also considered the following generalization: Suppose that we have an MDS code with three parity nodes, if we have a single erasure, using our codes, we can rebuild the erasure with rebuilding  ratio of $1/3$. What happens if we have two erasures? What is the rebuilding ratio in this case?
Our codes can achieve the optimal rebuilding ratio of $2/3$. In general, if we have $r \ge 3$ parity nodes and $e$ erasures happen, $1 \le e \le r$,  we will prove that the lower bound of repair bandwidth is $e/r$ (normalized by the size of the remaining array), and so is the rebuilding ratio. And the code we constructed achieves this lower bound for any $e$.

In summary, the main contribution of this paper is the first explicit construction of systematic $(n,k)$ MDS array codes for any constant $r=n-k$, which achieves optimal rebuilding ratio of $\frac{1}{r}$. Moreover, our codes achieve optimal rebuilding ratio of $\frac{e}{r}$ when $e$ systematic erasures occur, $1 \le e \le r$. The parity symbols are constructed by linear combinations of a set of information symbols, such that each information symbol is contained exactly once in each parity node. These codes have a variety of advantages: 1) they are systematic codes, hence it is easy to retrieve information;
2) they have high code rate $k/n$, which is commonly required in storage systems; 3) the encoding and decoding of the codes can be easily implemented (for $r=2$, the code uses finite field of size 3); 4) they match the lower bound of the ratio when rebuilding $e$ systematic nodes; 5) the rebuilding of a failed node requires simple computation and access to only $1/r$ of the data in each node (no linear combination of data); 6) they have \emph{optimal update}, namely, when an information element is updated, only $r+1$ elements in the array need update; and 7) they have optimal array size.

The remainder of the paper is organized as follows. Section \ref{sec3} constructs $(k+2,k)$ MDS array codes with optimal rebuilding ratio. Section \ref{sec2} gives formal definitions and some general observations on MDS array codes. Section \ref{code-duplication} introduces code duplication and thus generates $(k+2,k)$ MDS array codes for arbitrary number of columns. We discuss the size of the finite field needed for these constructions in Section \ref{section 5}. Decoding algorithms for erasures and errors are discussed in Section \ref{sec:dec}. Section \ref{generalization} generalizes the MDS code construction to arbitrary number of parity columns. These generalized codes have properties that are similar to the $(k+2,k)$ MDS array codes, likewise some of them has optimal rebuilding ratio. Rebuilding of multiple erasures and generalization of the rebuilding algorithms are presented in Section \ref{sec:multi}. Finally we provide concluding remarks in Section \ref{summary}.

%%%%%%%%%%%%%%%%%%%%%%%%%%%%%%%%%%%%%%%%%%%%%%%%%%%%%%%%%%%%%%%%%%%
%
%
\section{$(k+2,k)$ MDS array code constructions} \label{sec3}
%
%
%%%%%%%%%%%%%%%%%%%%%%%%%%%%%%%%%%%%%%%%%%%%%%%%%%%%%%%%%%%%%%%%%%%
In the rest of the paper, we are going to use $[i,j]$ to denote $\{i,i+1,\dots,j\}$ and $[i]$ to denote $\{1,2,\dots,i\}$, for integers $i \le j$. And denote the complement of a subset $X\subseteq M$ as $\overline{X}=M\backslash X$. For a matrix $A$, $A^T$ denotes the transpose of $A$.
 For a binary vector $v=(v_1,...,v_n)$ we denote by $\overline{v}=(v_1+1\mod 2 ,...,v_n+1\mod 2)$ its complement vector. The standard vector basis of dimension $m$ will be denoted as $\{e_i\}_{i=1}^m$ and the zero vector will be denoted as $e_0.$ For an integer $n$ denote by $S_n$ the set of permutations over the integers $[0,n-1]$, namely the symmetric group. For two functions $f,g$, denote their composition by $fg$ or $f \circ g$.

Recall that \emph{rebuilding ratio} is the average fraction of accesses in the surviving systematic and parity nodes while rebuilding one systematic node. More specific definition will be given in the next section.
In this section we give the construction of MDS array code with two parities and optimal rebuilding ratio $1/2$ for one erasure, which uses an optimal finite field of only size $3$.

We mentioned in the introduction that each $(k+2,k)$ MDS array code with optimal update can be constructed by defining the row and zigzag parities (proofs are given in Section \ref{sec2}). More specifically, the row parity corresponds to identity permutation in each systematic column, and the zigzag parity corresponds to a set of permutations $\{f_0,f_1,\dots,f_{k-1}\}$ for the systematic columns $\{0,1,\dots,k\}$. From the example in Figure \ref{fig:shapes}, we know that in order to get low rebuilding ratio, we need to find $f_0,...,f_{k-1}$ such that the row and zigzag sets used in rebuilding intersect as much as possible.
In addition, since each parity element is a linear combination of elements in its parity set, we need to define the coefficients of the linear combination such that the code is MDS. Noticing that all elements and all coefficients are from some finite field, we would like to choose the coefficients such that the finite field size is as small as possible. So our construction of the code includes two steps:
\begin{enumerate}
\item
Find zigzag permutations to minimize the ratio.
\item
Assign the coefficients such that the code is MDS.
\end{enumerate}

The following construction constructs a family of MDS array codes with $2$ parities using binary vectors.
From any set $T\subseteq \mathbb{F}_2^m$, $|T|=k$, we construct a $(k+2,k)$ MDS array code of size $2^m\times (k+2).$
We will show that some of these codes have the optimal ratio of $\frac{1}{2}$.

In this section all the calculations are done over $\mathbb{F}_2$. By abuse of notation we use $x\in [0,2^{m}-1]$ both to represent the integer and its binary representation. It will be clear from the context which meaning is in use.

\begin{cnstr}
Let $A=(a_{i,j})$ be an array of size $2^m\times k$ for some integers $k,m$ and $k < 2^m$. Let $T\subseteq \mathbb{F}^m_2$ be a set of vectors of size $k$ which does not contain the zero vector. For $v\in T$ we define the permutation $f_v:[0,2^m-1]\to [0,2^m-1] $ by $f_v(x)=x+v$, where $x$ is represented in its binary representation. One can check that this is actually a permutation. For example when $m=2,v=(1,0),x=3$, $$f_{(1,0)}(3)=3+(1,0)=(1,1)+(1,0)=(0,1)=1.$$
One can check that the permutation $f_v$ in vector notation is $[2, 3, 0, 1]$.
In addition, we define $X_v=\{x \in [0,2^m-1]:x\cdot v=0\}$ as the set of integers orthogonal to $v$. For example, $X_{(1,0)}=\{0,1\}$. The construction of the two parity columns is as follows:
The first parity is simply the row sums.
The second parity is the linear combination of elements in the zigzag set.
The zigzag sets $Z_0,...,Z_{2^m-1}$ are defined by the permutations $\{f_{v_j}:v_j\in T\}$ as $a_{i,j}\in Z_l$ if $f_{v_j}(i)=l.$ We will denote the permutation $f_{v_j}$ as $f_j$ and the set $X_{v_j}$ as $X_j$.
Assume column $j$ is erased, and define $S_r=\{a_{i,j}:i\in X_j\}$ and $S_z=\{a_{i,j}:i\notin X_j\}$. Rebuild the elements in $S_r$ by rows and the elements in $S_z$ by zigzags.
\label{cnstr1}
\end{cnstr}

\begin{thm}
\label{orthogonal-permutations}
Construct permutations $f_0,...,f_m$ and sets $X_0,...,X_m$ by the vectors $\{e_i\}_{i=0}^m$ as in Construction \ref{cnstr1}, where $X_0$ is modified to be $X_0=\{x\in \mathbb{F}_2^m:x\cdot (1,1,...,1)=0\}$. Then the corresponding $(m+3,m+1)$ code has \emph{optimal} ratio of $\frac{1}{2}.$
\end{thm}

Before proving the theorem, we first give an example. Actually, this example is the code in Figure \ref{fig:shapes} with more details.

\begin{xmpl}
Let $A$ be an array of size $4\times 3$. We construct a $(5,3)$ MDS array code for $A$ as in Theorem \ref{orthogonal-permutations} that accesses  $\frac{1}{2}$ of the remaining information in the array to rebuild any systematic node (see Figure \ref{fig2}).
\begin{figure*}[t]
	\centering
	\includegraphics[scale=0.65]{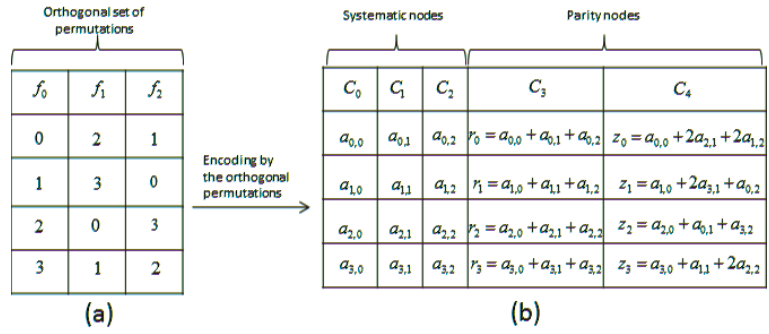}
\caption{$(a)$ The set of orthogonal permutations as in Theorem \ref{orthogonal-permutations} with sets $X_0=\{0,3\},X_1=\{0,1\},X_2=\{0,2\}$. $(b)$ A $(5,3)$ MDS array code generated by the orthogonal permutations. The first parity column $C_3$ is the row sum and the second parity column $C_4$ is generated by the zigzags. For example, zigzag $z_0$ contains the elements $a_{i,j}$ that satisfy $f_j(i)=0$.}
	\label{fig2}
\vspace{-0.5cm}
\end{figure*}
For example, $X_1=\{0,1\}$, and for rebuilding  of node $1$ (column $C_1$) we  access the elements $a_{0,0},a_{0,2},a_{1,0},a_{1,2}$, and the following four parity elements
\begin{align*}
&r_0=a_{0,0}+a_{0,1}+a_{0,2}\\
&r_1=a_{1,0}+a_{1,1}+a_{1,2}\\
&z_{f_1(2)}=z_0=a_{0,0}+2a_{2,1}+2a_{1,2}\\
&z_{f_1(3)}=z_1=a_{1,0}+2a_{3,1}+a_{0,2}.
\end{align*}
It is trivial to rebuild node $1$ from the accessed information. Note that each of the surviving node accesses exactly $\frac{1}{2}$ of its elements. It can be easily verified that the other systematic nodes can be rebuilt the same way. Rebuilding a parity node is easily done by accessing all the information elements.
\end{xmpl}

In order to prove Theorem \ref{orthogonal-permutations}, we first prove the following lemma. We use a binary vector to represent its corresponding systematic node. And define $|v\backslash u|= \sum_{i:v_i=1,u_i=0}1$ as the number of coordinates at which $v$ has a $1$ but $u$ has a $0$.

\begin{lem}
\label{lemma 3}
(i) For any $v,u \in T$, to rebuild node $v$, the number of accessed elements in node $u$ is
$$2^{m-1}+|f_v(X_v)\cap f_u(X_v)|.$$
(ii) For any $0 \neq v,u\in T$,
\begin{equation}
|f_v(X_v)\cap f_u(X_v)|=
\begin{cases}
|X_v|, & |v \backslash u|\mod 2 =0\\
0,     & |v \backslash u|\mod 2 =1.
\end{cases}
\end{equation}
\end{lem}

\begin{IEEEproof}
(i) In rebuilding of node $v$ we rebuild the elements in rows $X_v$ by rows, thus the row parity column accesses the values of the sum of rows $X_v$. Moreover, the surviving node $u$ also accesses its elements in rows $X_v$. Hence, by now $|X_v|=2^{m-1}$ elements are accessed. The elements of node $v$ in rows $\overline{X_v}$ are rebuilt by zigzags, thus the zigzag parity column accesses the values of the zigzags sums $\{z_{f_v(l)}:l\in \overline{X_v}\}$, and each surviving systematic node accesses the elements of these zigzags from its column, unless these elements are already included in the rebuilding by rows. The zigzag elements in $\{Z_{f_v(l)}:l\in \overline{X_v}\}$ of node $u$ are in rows $f_u^{-1}(f_v(\overline{X_v}))$, where $f_u^{-1}$ is the inverse function of $f_u$. Thus the extra elements node $u$ needs to access are  in rows $f_u^{-1}(f_v(\overline{X_v}))\backslash X_v.$ But,
\begin{eqnarray*}
&&|f_u^{-1}(f_v(\overline{X_v}))\backslash X_v| \\
&=&|\overline{f_u^{-1}(f_v(X_v))}\cap \overline{X_v}|\\
&=&|\overline{f_u^{-1}(f_v(X_v))\cup X_v}|\\
&=&2^m-|f_u^{-1}(f_v(X_v))\cup X_v|\\
&=&2^m-(|f_u^{-1}(f_v(X_v))|+|X_v|-|f_u^{-1}(f_v(X_v))\cap X_v|)\\
&=&|f_u^{-1}(f_v(X_v))\cap X_v|\\
&=&|f_v(X_v)\cap f_u(X_v)|,
\end{eqnarray*}
where we used the fact that $f_v,f_u$ are bijections, and $|X_v|=2^{m-1}.$

(ii) Consider the group $(\mathbb{F}_2^m,+)$. Recall that $f_v(X)=X+v=\{x+v:x\in X\}$.
The sets $f_v(X_v)=X_v+v$ and $f_u(X_v)=X_v+u$ are cosets of the subgroup $X_v=\{w\in \mathbb{F}_2^m:w\cdot v=0\}$, and they are either identical or disjoint. Moreover, they are identical iff $v-u\in X_v$, namely $(v-u)\cdot v=\sum_{i:v_i=1,u_i=0}1 \equiv 0 \mod 2.$ However, by definition $|v \backslash u| \equiv \sum_{i:v_i=1,u_i=0}1\mod 2$, and the result follows.
\end{IEEEproof}

Let $\{f_0,...,f_{k-1}\}$ be a set of permutations over the set $[0,2^{m}-1]$ with associated subsets $X_0,...,X_{k-1}\subseteq [0,2^{m}-1]$, where each $|X_i|=2^{m-1}$. We say that this set is a set of \emph{orthogonal permutations} if for any $i,j\in [0,k-1]$, $$\frac{|f_i(X_i)\cap f_j(X_i)|}{2^{m-1}}=\delta_{i,j},$$ where $\delta_{i,j}$ is the Kronecker delta. For a set of orthogonal permutations, in order to rebuild any systematic node, only $2^{m-1}$ elements are accessed from each surviving systematic node by Lemma \ref{lemma 3}. And only $2^{m-1}$ elements are accessed from each parity node, too. Hence codes generated by orthogonal permutations has optimal rebuilding ratio $1/2$.
Now we are ready to prove Theorem \ref{orthogonal-permutations}.

\begin{IEEEproof}[Proof of Theorem \ref{orthogonal-permutations}]
Since $|{e_i} \backslash {e_j}|=1$ for any $i\neq j\neq 0$, we get by lemma \ref{lemma 3} $$f_i(X_i)\cap f_j(X_i)=\emptyset.$$
Now consider $e_i$ and $e_0$, for $i \neq 0$. Note that
$f_i(X_i)=\{x+e_i:x \cdot e_i=0\} = \{y:y \cdot e_i=1\},$
so
$$f_i(X_i)\cap f_0(X_i)=\{y:y \cdot e_i=1\} \cap \{x:x \cdot e_i=0\}=\emptyset.$$
Similarly,
$f_i(X_0)=\{x+e_i:x \cdot (1,1,\dots,1)=0\}
=\{y:y \cdot (1,1,\cdots,1)=1\},$
and
\begin{eqnarray*}
&& f_0(X_0) \cap f_i(X_0) \\
&=&\{x:x \cdot (1,\cdots,1)=0\} \cap \{y:y \cdot (1,\cdots,1)=1\}\\
&=&\emptyset.
\end{eqnarray*}
Hence the permutations $f_0,\dots,f_m$ are orthogonal permutations, and the ratio is $1/2$.
\end{IEEEproof}

Note that the optimal code can be shortened by removing some systematic columns and still retain an optimal ratio, i.e., for any $k\leq m+1$ we have a code with optimal rebuilding.

Having found the set of orthogonal permutations, we need to specify the coefficients in the parities such that the code is MDS.

Consider the $(m+3,m+1)$ code $\mathcal{C}$ constructed by Theorem \ref{orthogonal-permutations} and the vectors $\{e_i\}_{i=0}^{m}$.
Let $\mathbb{F}$ be the finite field we use.
Let the information in row $i$, column $j$ be $a_{i,j} \in \mathbb{F}$. Let its row and zigzag coefficients be $\alpha_{i,j},\beta_{i,j} \in \mathbb{F}$.
For a row set $R_u=\{a_{u,0},a_{u,1},\dots,a_{u,m}\}$, the row parity is $r_u=\sum_{j=0}^{m} \alpha_{u,j}a_{u,j}$. For a zigzag set $Z_u=\{a_{u, 0},a_{u+e_1,1},\dots,a_{u+e_m,m}\}$, the zigzag parity is
$z_u = \sum_{j=0}^{m}\beta_{u+e_j,j}a_{u+e_j,j}$.

Recall that the $(m+3,m+1)$ code is MDS iff we can recover the information from up to $2$ columns erasures.
It is clear that none of the coefficients $\alpha_{i,j},\beta_{i,j}$ can be zero. Moreover, if we assign all the coefficients as $\alpha_{i,j}=\beta_{i,j}=1$ we get that in an erasure of two systematic columns the set of equations derived from the parity columns are linearly dependent and thus not solvable (the sum of the equations from the row parity and the sum of those from the zigzag parity will both be the sum of the entire information array). Therefore the coefficients need to be from a field with more than $1$ nonzero element, thus a field of size at least $3$ is necessary. The construction below surprisingly shows that in fact $\mathbb{F}_3$ is sufficient.

\begin{cnstr} \label{cons3}
For the code $\mathcal{C}$ in Theorem \ref{orthogonal-permutations} over $\mathbb{F}_3$, define $u_j=\sum_{l=0}^j e_l$ for $0 \le j \le m$. Assign row coefficients as $\alpha_{i,j}=1$ for all $i,j$, and zigzag coefficients as
$$\beta_{i,j}=2^{i \cdot u_j }$$
where $i=(i_1,\dots,i_m)$ is represented in binary and the calculation in the exponent is done over $\mathbb{F}_2$.
\end{cnstr}

The coefficients in Figure \ref{fig2} are assigned by Construction \ref{cons3}. The following theorem shows that the code is MDS.

\begin{thm} \label{thm0506}
Construction \ref{cons3} is an $(m+3,m+1)$ MDS code with optimal finite field size of $3$.
%\vspace{-0.1cm}
\end{thm}
\begin{IEEEproof}
It is easy to see that if at least one of the two erased columns is a parity column then we can recover the information. Hence we only need to show that we can recover from any erasure of two systematic columns. In an erasure of two systematic columns $i,j\in [0,m],i<j$, we access the entire remaining information in the array. For $r\in [0,2^m-1]$ set $r'=r+e_i+e_j$, and recall that $a_{r,i}\in Z_{l}$ iff $l=r+e_i$, thus $a_{r,i},a_{r',j}\in Z_{r+e_i}$ and $a_{r,j},a_{r',i}\in Z_{r+e_j}.$ From the two parity columns we need to solve the following equations (for some $y_1,y_2,y_3,y_4 \in \mathbb{F}_3$)
$$
\left[\begin{array}{cccc}
1 & 1 & 0 & 0 \\
0 & 0 & 1 & 1 \\
\beta_{r,i} & 0 & 0 & \beta_{r',j} \\
0 & \beta_{r,j} & \beta_{r',i} & 0
\end{array} \right]
\left[ \begin{array}{c}
a_{r,i} \\
a_{r,j} \\
a_{r',i} \\
a_{r',j} \\
\end{array} \right]
= \left[ \begin{array}{c}
y_1 \\
y_2 \\
y_3 \\
y_4
\end{array} \right].
$$
This set of equations is solvable iff
\begin{equation} \label{eq1}
\beta_{r,i} \beta_{r',i} \neq \beta_{r,j}  \beta_{r',j}.
\end{equation}
Note that the multiplicative group of $\mathbb{F}_3 \backslash{0}$  is isomorphic to the additive group of $\mathbb{F}_2$, hence multiplying two elements in $\mathbb{F}_3 \backslash{0}$ is equivalent to summing up their exponent in $\mathbb{F}_2$ when they are represented as a power of the primitive element of the field $\mathbb{F}_3$. For columns $0 \le i<j \le m$ and rows $r,r'$ defined above, we have
$$\beta_{r,i} \beta_{r',i}=2^{r \cdot u_i  +  r'\cdot u_i }=2^{ (r+ r')\cdot u_i}=2^{(e_i+e_j)\cdot\sum_{l=0}^{i}e_l}=2^{e_i^2}=2.$$
However in the same manner we derive that
$$\beta_{r,j}  \beta_{r',j}= 2^{(r+ r') \cdot u_j} = 2^{(e_i+e_j)\cdot \sum_{l=0}^{j}e_l} =  2^{e_i^2+e_j^2}=2^0=1.$$
Hence \eqref{eq1} is satisfied and the code is  MDS.
\end{IEEEproof}

\textbf{Remark:}
The above proof shows that $\beta_{r,i} \neq \beta_{r',i}$, and $\beta_{r,j}=\beta_{r',j}$ for $i<j$. And \eqref{eq1} is a necessary and sufficient condition for a MDS code for vectors $v_i \neq v_j$.

In addition to optimal ratio and optimal field size, we will show in the next section that the code in Theorem \ref{orthogonal-permutations} is also of optimal array size, namely, it has the maximum number of columns, given the number of rows.

%%%%%%%%%%%%%%%%%%%%%%%%%%%%%%%%%%%%%%%%%%%%%%%%%%%%%%%%%%%%%%%%%%%
%
%
\section{Formal Problem Settings and Constructions} \label{sec2}
%
%
%%%%%%%%%%%%%%%%%%%%%%%%%%%%%%%%%%%%%%%%%%%%%%%%%%%%%%%%%%%%%%%%%%%
In this section, we first give some observations of an arbitrary MDS array code with optimal update. Then we prove some properties and give some examples of our code in Construction \ref{cnstr1}.

Let us define an MDS array code with 2 parities.
Let $A=(a_{i,j})$ be an array of size $p\times k$ over a finite field $\mathbb{F}$, where $i\in [0,p-1],j\in [0,k-1]$, and each of its entry is an information element. We add to the array two parity columns and obtain an $(n=k+2,k)$ MDS code of array size $p \times n$. Each element in these parity columns is a linear combination of elements from $A$. More specifically, let the two parity columns be $C_k=(r_0,r_1,...,r_{p-1})^T$ and $C_{k+1}=(z_0,z_1...,z_{p-1})^T$.
Let $R=\{R_0,R_1,...,R_{p-1}\}$ and $Z=\{Z_0,Z_1,...,Z_{p-1}\}$ be two sets such that $R_l,Z_l$ are subsets of elements in $A$ for all $l \in [0,p-1]$.
Then for all $l \in [0,p-1]$, define $r_l=\sum_{a\in R_l}\alpha_aa\text{ and } z_l=\sum_{a\in Z_l}\beta_aa$, for some sets of coefficients $\{\alpha_a\},\{\beta_a\}\subseteq \mathbb{F}$.
We call $R$ and $Z$ as the sets that generate the parity columns.

We assume the code has optimal update, meaning that only $3$ elements in the code are updated when an information element is updated. Under this assumption, the following theorem characterizes the sets $R$ and $Z$.

\begin{thm}
For an $(k+2,k)$ MDS code with optimal update,
the sets $R$ and $Z$ are partitions of $A$ into $p$ equally sized sets of size $k$, where each set in $R$ or $Z$ contains exactly one element from each column.
\end{thm}
\begin{IEEEproof}
Since the code is a $(k+2,k)$ MDS code, each information element should appear at least once in each parity column $C_{k},C_{k+1}.$ However, since the code has optimal update, each element appears exactly once in each parity column.

Let $X\in R$, note that if $X$ contains two entries of $A$ from the systematic column $C_i$, $i \in [0,k-1]$, then rebuilding is impossible if columns $C_i$ and $C_{k+1}$ are erased. Thus $X$ contains at most one entry from each column, therefore $|X|\leq k.$ However each element of $A$ appears exactly once in each parity column, thus if $|X|<k$, $X$ $\in$ $R$, there is $Y\in R$, with $|Y|>k$, which leads to a contradiction. Therefore, $|X|=k$ for all $X \in R$. As each information element appears exactly once in the first parity column, $R=\{R_0,\dots,R_{p-1}\}$ is a partition of $A$ into $p$ equally sized sets of size $k$. Similar proof holds for the sets $Z=\{Z_0,\dots,Z_{p-1}\}$.
\end{IEEEproof}

By the above theorem, for the $j$-th systematic column $(a_0,\dots,a_{p-1})^T$, its $p$ elements are contained in $p$ distinct sets $R_l$, $l \in [0,p-1]$. In other words, the membership of the $j$-th column's elements in the sets $\{R_l\}$ defines a permutation $g_j:[0,p-1] \to [0,p-1]$, such that  $g_j(i)=l$ iff $a_i \in R_{l}$. Similarly, we can define a permutation $f_j$ corresponding to the second parity column, where $f_j(i)=l$ iff $a_i \in Z_{l}$. For example, in Figure \ref{fig:shapes} each systematic column corresponds to a permutation of the four symbols.

Observing that there is no importance of the elements' ordering in each column, w.l.o.g. we can assume that the first parity column contains the sum of each row of $A$ and $g_j$'s correspond to identity permutations, i.e. $r_i=\sum_{j=0}^{k-1} \alpha_{i,j}a_{i,j}$ for some coefficients $\{\alpha_{i,j}\}$. We refer to the first and the second parity columns as the row parity and the zigzag parity respectively, likewise $R_l$ and $Z_l$, $l \in [0,p-1]$, are referred to as row sets and zigzag sets respectively. We will call $f_j$, $j \in [0,k-1]$, zigzag permutations. By assuming that the first parity column contains the row sums, we want to (1) find zigzag permutations to minimize the rebuilding ratio; and (2) assign the coefficients such that the code is MDS.

First we show that any set of zigzag sets $Z=\{Z_0,...,Z_{p-1}\}$ defines a $(k+2,k)$ MDS array code over a field $\mathbb{F}$ large enough.

\begin{thm}
\label{zigzag-sets}
Let $A=(a_{i,j})$ be an array of size $p\times k$ and the zigzag sets be $Z=\{Z_0,...,Z_{p-1}\}$, then there exists a $(k+2,k)$ MDS array code for $A$ with $Z$ as its zigzag sets over the field $\mathbb{F}$ of size greater than $p(k-1)+1$.
\end{thm}

The proof is shown in Appendix \ref{app1}. The above theorem states that there exist coefficients such that the code is MDS, and thus we will focus first on finding proper zigzag permutations $\{f_j\}$. The idea behind choosing the zigzag sets is as follows: assume a systematic column $(a_0,a_1,...,a_{p-1})^T$ is erased. Each element $a_i$  is contained in exactly one row set and one zigzag set. For rebuilding of element $a_i$, access the parity of its row set or zigzag set. Moreover access the values of the remaining elements in that set, except $a_i$. We say that an element $a_i$ is rebuilt by a row (zigzag) if the parity of its row set (zigzag set) is accessed. For example, in Figure \ref{fig:shapes} supposing column $1$ is erased, one can access the shaded elements and rebuild its first two elements by rows, and the rest by zigzags. The set $\mathbf{S}=\{S_0,S_1,...,S_{p-1}\}$ is called a rebuilding set for column $(a_0,a_1,...,a_{p-1})^T$ if for each $i$, $S_i\in R\cup Z\text{ and }a_i\in S_i$. In order to minimize the number of accesses to rebuild the erased column, we need to minimize  the size of \begin{equation}
\label{eq:77}
|\cup_{i=0}^{p-1} S_i|,
\end{equation}
which is equivalent to maximizing the number of intersections between the sets $\{S_i\}_{i=0}^{p-1}$. More specifically, the intersections between the row sets in $\mathbf{S}$ and the zigzag sets in $\mathbf{S}$.

For a $(k+2,k)$ MDS code $\cC$ with $p$ rows define the \emph{rebuilding ratio} $R(\mathcal{C})$ as
the average fraction of accesses in the surviving systematic and parity nodes while rebuilding one systematic node, i.e.,
$$R(\mathcal{C}) =\frac{\sum_{j} \min_{{S_0,\dots,S_{p-1}} \text{ rebuilds } j}|\cup_{i=0}^{p-1} S_i|}{p(k+1)k}.$$
Notice that in the two parity nodes, we access $p$ elements because each erased element must be rebuilt either by row or by zigzag. And $\cup_{i=0}^{p-1} S_i$ contains $p$ elements in the erased column. Thus the above expression is exactly the rebuilding ratio.
Define the \emph{ratio function} for all $(k+2,k)$ MDS codes with $p$ rows as
$$R(k)=\min_{\cC}{R(\cC)},$$
which is the minimal average portion of the array needed to be accessed in order to rebuild one erased column.

\begin{thm}
\label{monotone function}
$R(k)$ is no less than $\frac{1}{2}$ and is a monotone nondecreasing function.
\end{thm}

The proof is given in Appendix \ref{app2}. For example, the code in Figure \ref{fig2} achieves the lower bound of ratio $1/2$, and therefore $R(3)=1/2$. Moreover, We will see in Corollary \ref{thm2} that $R(k)$ is almost $1/2$ for all $k$ and  $p=2^m$, where $m$ is large enough.

So far we have discussed the characteristics of an arbitrary MDS array code with optimal update. Next, let us look at our code in Construction \ref{cnstr1}.

Recall that by Theorem \ref{zigzag-sets} this code can be an MDS code over a field large enough.
The ratio of the constructed code will be proportional to the size of the union of the elements in the rebuilding set in \eqref{eq:77}.
The following theorem gives the ratio for Construction \ref{cnstr1} and can be easily derived from Lemma \ref{lemma 3} part (i).
\begin{thm}
\label{th:123}
The code described in Construction \ref{cnstr1} and generated by the vectors $v_0,v_1,...,v_{k-1}$ is a $(k+2,k)$  MDS array code with ratio
\begin{equation}
R=\frac{1}{2}+\frac{\sum_{i=0}^{k-1}\sum_{j\neq i}|f_i(X_i)\cap f_j(X_i)|}{2^mk(k+1)}.\label{eq:345}
\end{equation}
\end{thm}

Next we show the optimal code in Theorem \ref{orthogonal-permutations} is optimal in size, namely, it has the maximum number of columns given the number of rows.

\begin{thm}\label{thm:size}
Let $F$ be an orthogonal set of  permutations over the integers $[0,2^m-1]$, then the size of $F$ is at most $m+1.$
\end{thm}
\begin{IEEEproof}
We will prove it by induction on $m$. For $m=0$ there is nothing to prove. Let $F=\{f_0,f_1,...,f_{k-1}\}$ be a set of orthogonal permutations over the set $[0,2^m-1].$ We only need to show that $|F|=k\leq m+1.$ It is trivial to see that for any permutations $g,h$ on $[0,2^m-1]$, the set $hFg=\{hf_0g,hf_1g,...,hf_{k-1}g\}$ is also a set of orthogonal permutations with sets $g^{-1}(X_0),g^{-1}(X_1),...,g^{-1}(X_{k-1}).$ Thus w.l.o.g. we can assume that $f_0$ is the identity permutation and $X_0=[0,2^{m-1}-1]$. From the orthogonality we get that
$$\cup_{i=1}^{k-1} f_i(X_0)=\overline{X_0}=[2^{m-1},2^{m}-1].$$ We claim that for any $i\neq 0,|X_i\cap X_0|=\frac{|X_0|}{2}=2^{m-2}.$ Assume the contrary, thus w.l.o.g we can assume that $|X_i\cap X_0|>2^{m-2}$, otherwise $|X_i\cap \overline{X_0}|>2^{m-2}.$ For  any $j\neq i\neq 0$ we get that
\begin{equation}
\label{eq:444}
f_j(X_i\cap X_0), f_i(X_i\cap X_0)\subseteq \overline{X_0},
\end{equation}
\begin{equation}
\label{eq:445}
|f_j(X_i\cap X_0)|=|f_i(X_i\cap X_0)|>2^{m-2}=\frac{|\overline{X_0}|}{2}.
\end{equation}
From equations (\ref{eq:444}) and (\ref{eq:445}) we conclude that $f_j(X_i\cap X_0)\cap f_i(X_i\cap X_0)\neq \emptyset$, which contradicts the orthogonality property. Define the set of  permutations $F^*=\{f_i^*\}_{i=1}^{k-1}$ over the set of integers $[0,2^{m-1}-1]$ by $f_i^*(x)=f_i(x)-2^{m-1}$, which is a set of orthogonal permutations with sets $\{X_i\cap X_0\}_{i=1}^{k-1}$. By induction $k-1 \leq m$ and the result follows.
\end{IEEEproof}

The above theorem implies that the number of rows has to be exponential in the number of columns in any systematic code with optimal ratio and optimal update. Notice that the code in Theorem \ref{orthogonal-permutations} achieves the \emph{maximum} possible number of columns, $m+1$. Besides, an exponential number of rows is still practical in storage systems, since they are composed of dozens of nodes (disks) each of which has size in an order of gigabytes. In addition, increasing the number of columns can be done using duplication (Theorem \ref{lem13}) or a larger set of vectors (the following example) with a cost of a small increase in the ratio.

\begin{xmpl} \label{xmpl1}
Let $T=\{v \in \mathbb{F}_2^m:\|v\|_1=3\}$ be the set of vectors with weight 3 and length $m$. Notice that $|T|=\binom{m}{3}$. Construct the code $\cC$ by $T$ according to Construction \ref{cnstr1}.
Given $v \in T$, $|\{u \in T: |v\backslash u|=3\}|= \binom{m-3}{3}$, which is the number of vectors with 1's in different positions as $v$. Similarly,
$|\{u \in T: |v\backslash u|=2\}|= 3\binom{m-3}{2}$ and
$|\{u \in T: |v\backslash u|=1\}|= 3(m-3)$. By Theorem \ref{th:123} and Lemma \ref{lemma 3}, for large $m$ the ratio is
$$\frac{1}{2}+\frac{2^{m-1}\binom{m}{3}3\binom{m-3}{2}}{2^m\binom{m}{3} (\binom{m}{3}+1)} \approx \frac{1}{2} + \frac{9}{2m}.$$
\end{xmpl}

Note that this code reaches the lower bound of the ratio as $m$ tends to infinity, and has $O(m^3)$ columns.

%%%%%%%%%%%%%%%%%%%%%%%%%%%%%%%%%%%%%%%%%%%%%%%%%%%%%%%%%%%%%%%%%%%
%
%
\section{Code Duplication}
\label{code-duplication}
%
%
%%%%%%%%%%%%%%%%%%%%%%%%%%%%%%%%%%%%%%%%%%%%%%%%%%%%%%%%%%%%%%%%%%%
In this section, we are going to duplicate the code to increase the number of columns in the constructed $(k+2,k)$ MDS codes, such that $k$ does not depend on the number of rows, and ratio is approximately $\frac{1}{2}$. Then we will show the optimality of the duplication code based on the standard basis.

Let $\mathcal{C}$ be a $(k+2,k)$ array code with $p$ rows, where the zigzag sets $\{Z_l\}_{l=0}^{p-1}$ are defined by the set of permutations $\{f_i\}_{i=0}^{k-1}$ on $[0,p-1]$. For an integer $s$, an $s$-\emph{duplication code} $\mathcal{C}'$ is an $(sk+2,sk)$ MDS code with zigzag permutations defined by duplicating the $k$ permutations $s$ times each, and the first parity column is the row sums. In order to make the code MDS, the coefficients in the parities may be different from the code $\cC$. For an $s$-duplication code, denote the column corresponding to the  $t$-th $f_j$ as column $j^{(t)}$, $0 \le t \le s-1$. Call the columns $\{j^{(t)}: j \in [0,k-1]\}$ the $t$-th copy of the original code. An example of a $2$-duplication of the code in Figure \ref{fig2} is illustrated in Figure \ref{fig:duplication}.

\begin{figure*}
	\centering
		\includegraphics[scale=.75]{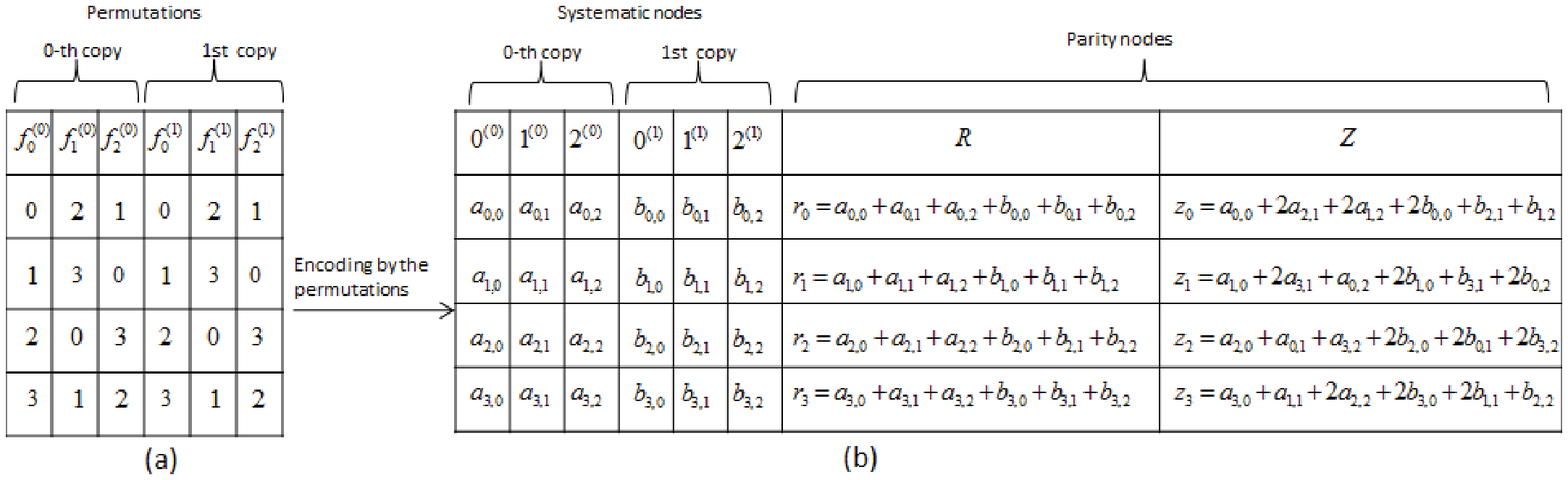}
	\caption{A $2$-duplication of the code in Figure \ref{fig2}. The code has $6$ information nodes and $2$ parity nodes. The ratio is $4/7$.}
	\label{fig:duplication}
\end{figure*}

\begin{thm} \label{lem13}
If a $(k+2,k)$ code $\mathcal{C}$ has ratio $R(\mathcal{C})$, then its $s$-duplication code $\mathcal{C}'$ has ratio  $R(\mathcal{C})(1+\frac{s-1}{sk+1})$.
\end{thm}

\begin{IEEEproof}
We propose a rebuilding algorithm for $\mathcal{C}'$ with ratio of $R(\mathcal{C})(1+\frac{s-1}{sk+1})$, which will be shown to be optimal.

Suppose in the optimal rebuilding algorithm of $\mathcal{C}$, for column $i$,
elements of rows $J=\{j_1,j_2,\dots,j_u\}$ are rebuilt by zigzags, and the rest by rows. In $\mathcal{C}'$, all the $s$ columns corresponding to $f_i$ are rebuilt in the same way: the elements in rows $J$ are rebuilt by zigzags.

W.l.o.g. assume column $i^{(0)}$ is erased. Since column $i^{(t)}$, $t \in [1,s-1]$ corresponds to the same zigzag permutation as the erased column, for the erased element in the $l$-th row, no matter if it is rebuilt by row or by zigzag, we have to access the element in the $l$-th row and column $i^{(t)}$ (e.g. permutations $f_0^{(0)},f_0^{(1)}$ and the corresponding columns $0^{(0)},0^{(1)}$ in Figure \ref{fig:duplication}). Hence all the elements in column $i^{(t)}$ must be accessed. Moreover, the optimal way to access the other surviving columns can not be better than the optimal way to rebuild in the code $\cC$. Thus the proposed algorithm has optimal rebuilding ratio.

When column $i^{(0)}$ is erased, the average (over all $i \in [0,k-1]$) of the number of elements needed to be accessed in columns $l^{(t)}$, for all $l \in [0,k-1], l \neq i$ and $t \in [0,s-1]$ is
$$R(\cC)p(k+1)-p.$$
Here the term $-p$ corresponds to the access of the parity nodes in $\cC$. Moreover, we need to access all the elements in columns $i^{(t)}, 0<t \le s-1$, and access $p$ elements in the two parity columns. Therefore, the rebuilding ratio is
\begin{eqnarray*}
 R(\mathcal{C}')
&= & \frac{s(R(\mathcal{C})p(k+1)-p)+(s-1)p+p}{p(sk+1)} \\
&=& R(\mathcal{C})\frac{s(k+1)}{sk+1} \\
&=& R(\mathcal{C})(1+\frac{s-1}{sk+1})
\end{eqnarray*}
and the proof is completed.
\end{IEEEproof}

Theorem \ref{lem13} gives us the ratio of the $s$- duplication of a code $\cC$ as a function of its ratio $R(\cC)$.
As a result, for the optimal-ratio code in Theorem \ref{orthogonal-permutations}, the ratio of its duplication code is slightly more than $1/2$, as the following corollary suggests.

\begin{cor} \label{thm2}
The $s$-duplication of the code in Theorem \ref{orthogonal-permutations} has ratio $\frac{1}{2}(1+\frac{s-1}{s(m+1)+1})$, which is $\frac{1}{2}+\frac{1}{2(m+1)}$ for large $s$.
\end{cor}
%For example, when $k=10$, the array size is $1024 \times (11s+2)$, and the repair rate is $0.522$ and $0.537$, for $s=2$ and $s=6$, respectively.

For example, we can rebuild the column $1^{(0)}$ in Figure \ref{fig:duplication} by accessing the elements in rows $\{0,1\}$ and in columns $0^{(0)},2^{(0)},0^{(1)},2^{(1)},R,Z$, and all the elements in column $1^{(1)}$.
The rebuilding  ratio for this code is $4/7$.

Using duplication we can have \emph{arbitrarily large number of columns}, independent of the number of rows. Moreover the above corollary shows that it also has an almost optimal ratio.

Next we will show that if we restrict ourselves to codes constructed using Construction \ref{cnstr1} and duplication, the code using the standard basis and duplication has optimal asymptotic rate.

In order to show that, we define a related graph.
Define the directed graph $D_m=D_m(V,E)$ as $V=\{w \in \mathbb{F}_2^m:w\neq 0\}$, and $E=\{(w_1,w_2):|w_2 \backslash w_1|=1 \mod 2\}$. Hence the vertices are the nonzero binary vectors of length $m$, and there is a directed edge from $w_1$ to $w_2$ if $|w_2\backslash w_1|$ is odd size. %We also view the vertices as binary vectors of length $m$ where each vertex corresponds to its indicator vector.
 From any induced subgraph $H$ of $D_m$, we construct the code $\cC(H)$ from the vertices of $H$ using Construction \ref{cnstr1}.
By Lemma \ref{lemma 3} we know that a directed edge from $w_1$ to $w_2$ in $H$ means $f_{w_2}(X_{w_2}) \cap f_{w_1} (X_{w_2})= \emptyset$, so only half of the information from the column corresponding to $w_1$ is accessed while rebuilding the column corresponding to $w_2$. For a directed graph $D=D(V,E)$, let $S$ and $T$ be two disjoint subsets of its vertices. We define the density of the set $S$ to be $d_S=\frac{E_S}{|S|^2}$ and the density between $S$ and $T$ to be $d_{S,T}=\frac{E_{S,T}}{2|S||T|}$, where $E_S$ is the number of edges with both of its endpoints in $S$, and $E_{S,T}$ is the number of  edges incident with a vertex in $S$ and a vertex in $T$. The following theorem shows that the asymptotic ratio of any code constructed using Construction \ref{cnstr1} and duplication is a function of the density of the corresponding graph $H$.

\begin{theorem}
\label{th:34}
Let $H$ be an induced subgraph of $D_m$.
Let $\cC_s(H)$ be the $s$-duplication of the code constructed using the vertices of $H$ and Construction \ref{cnstr1}.
Then the asymptotic ratio of $\cC_s(H)$  is
$$\lim_{s \to \infty}R(\cC_s(H))=1-\frac{d_H}{2}$$
\end{theorem}
\begin{IEEEproof}
Let the set of vertices and edges of $H$ be $V(H)=\{v_i\}$ and $E(H)$ respectively. Denote by $v_i^l$, $v_i \in V(H), l \in [0,s-1]$, the $l$-th copy of the column corresponding to the vector $v_i$.
In the rebuilding of column $v_i^{l},l\in [0,s-1]$ each remaining systematic column $v_j^k,k\in[0,s-1]$, needs to access all of its $2^m$ elements unless $|v_i\backslash v_j|$ is odd, and in that case it only has to access $2^{m-1}$ elements. Hence the total amount of accessed information for rebuilding this column is $$(s|V(H)|-1)2^m-\deg^+(v_i)s2^{m-1},$$
where $\deg^+$ is the indegree of $v_i$ in the induced subgraph $H$. Averaging over all the columns in $\cC_s(H)$ we get  the ratio:
\begin{align}
 &R(\cC_s(H)) \nonumber\\
=&\frac{\sum_{v_i^l\in \cC_s(H)}(s|V(H)|-1)2^m-\deg^+(v_i)s2^{m-1}}{s|V(H)|(s|V(H)|+1)2^m}\nonumber\\
=&\frac{s|V(H)|(s|V(H)|-1)2^m-s^2\sum_{v_i\in V(H)}\deg^+(v_i)2^{m-1}}{s|V(H)|(s|V(H)|+1)2^m}\nonumber\\
=&\frac{s|V(H)|(s|V(H)|-1)2^m-s^2|E(H)|2^{m-1}}{s|V(H)|(s|V(H)|+1)2^m}.\nonumber
\end{align}
Hence $$\lim_{s \to \infty}R(\cC_s(H))=1-\frac{|E(H)|}{2|V(H)|^2}=1-\frac{d_H}{2}.$$

\end{IEEEproof}

We conclude from Theorem \ref{th:34} that the asymptotic ratio of any code using duplication and a set of binary vectors $\{v_i\}$ is a function of the density of the corresponding induced subgraph of $D_m$ with $\{v_i\}$ as its vertices. Hence the induced subgraph of $D_m$ with maximal density corresponds to the code with optimal asymptotic ratio. It is easy to check that the induced subgraph with its vertices as the standard basis $\{e_i\}_{i=1}^m$ has density $\frac{m-1}{m}$. In fact this is the maximal possible density among all the induced subgraph as Theorem \ref{thm:opt rate} suggests, but in order to show it we will need the following technical lemma.

\begin{lem}
\label{density lemma}
Let $D=D(V,E)$ be a directed graph and $S, T$ be a partition of $V$, i.e., $S\cap T=\emptyset, S\cup T=V$, then
$$d_V\leq \max \{d_S,d_T,d_{S,T}\}$$
\end{lem}

\begin{IEEEproof}
Note that $d_V=\frac{|S|^2d_S+|T|^2d_T+2|S||T|d_{S,T}}{|V|^2}.$ W.l.o.g assume that $d_S\geq d_T$ therefore if $d_S\geq D_{S,T}$,
\begin{align*}
d_V&=\frac{|S|^2d_S+|T|^2d_T+2|S||T|d_{S,T}}{|V|^2}\\
&\leq \frac{|S|^2d_S+|T|^2d_S-|T|^2d_S+|T|^2d_T+2|S||T|d_{S}}{|V|^2}\\
&=\frac{d_S(|S|+|T|)^2-|T|^2(d_S-d_T)}{|V|^2}\\
&\leq d_S.
\end{align*}
If $d_{S,T}\geq\max{\{d_S,d_T\}}$ then,
\begin{align*}
d_V&=\frac{|S|^2d_S+|T|^2d_T+2|S||T|d_{S,T}}{|V|^2}\\
&\leq \frac{|S|^2d_{S,T}+|T|^2d_{S,T}+2|S||T|d_{S,T}}{|V|^2}\\
&=d_{S,T}
\end{align*}
and the result follows.
\end{IEEEproof}

Now we are ready to prove the optimality of the duplication of the code using standard basis, if we assume that the number of copies $s$ tends to infinity.

\begin{thm} \label{thm:opt rate}
For any induced subgraph $H$ of $D_m, d_H\leq \frac{m-1}{m}$.
So the optimal asymptotic ratio among all codes constructed using duplication and Construction \ref{cnstr1} is
$\frac{1}{2}(1+\frac{1}{m})$ and is achieved using the standard basis.
\end{thm}
\begin{IEEEproof}
We say that a binary vector is an even (odd) vector if it has an even (odd) weight.
For two binary vectors $w_1,w_2$, $|w_2 \backslash w_1|$ being odd is equivalent to
$$1=w_2 \cdot \overline{w_1}= w_2\cdot((1,...,1)+w_1)=\|w_2\|_1+w_2\cdot w_1.$$
Hence, one can check that when $w_1,w_2$ have the same parity, there are either no edges or $2$ edges between them. Moreover, when their parities are different, there is exactly one edge between the two vertices.

When $m=1,$ the graph $D_1$ has only one vertex and the only nonempty induced subgraph is itself. $d_H=d_{D_1}=0=\frac{m-1}{m}$. When $m=2$, the graph $D_2$ has three vertices and one can check that the induced subgraph with maximum density contains $w_1=(1,0), w_2=(0,1)$, and the density is $1/2=(m-1)/m$.

For $m>2$, assume to the contrary that there exists a subgraph of $D_m$ with density higher than $\frac{m-1}{m}$.
Let $H$ be the smallest subgraph of $D_m$ (with respect to the number of vertices) among the subgraphs of $D_m$ with maximal density. Hence for any subset of vertices $S\subsetneq V(H)$, we have $d_S<d_H$. Therefore from Lemma \ref{density lemma} we conclude that for any partition $S,T$ of $V(H)$, $d_H\leq d_{S,T}.$
If $H$ contains both even and odd vectors, denote by $S$ and $T$ the set of even and odd vectors of $H$ respectively.
Since between any even and any odd vertex there is exactly one directed edge we get that
$d_H\leq d_{S,T}=\frac{1}{2}$. However $$\frac{1}{2}<\frac{m-1}{m}<d_H,$$ and we get a contradiction. Thus $H$ contains only odd vectors or even vectors.

Let $V(H)=\{v_1,...,v_k\}$. If this set of vectors is independent then $k\leq m$ and the outgoing degree for each vertex $v_i$ is at most $k-1$ hence $d_H=\frac{E(H)}{|V(H)|^2}\leq \frac{k(k-1)}{k^2}\leq \frac{m-1}{m}$ and we get a contradiction. Hence assume that the dimension of the subspace spanned by these vectors in $\mathbb{F}_2^m$ is $l<k$ where $v_1,v_2,...v_l$ are basis for it. Define $S=\{v_1,...v_l\},T=\{v_{l+1},...,v_k\}$. The following two cases show that the density can not be higher than $\frac{m-1}{m}$.

%\begin{enumerate}
%	\item
{\bf $H$ contains only odd vectors:}  Let $u\in T$. Since $u\in \spun\{S\}$ there is at least one $v\in S$ such that
$u\cdot v\neq 0$ and thus $(u,v),(v,u)\notin E(H)$, therefore the number of directed edges between $u$ and $S$ is at most $2(l-1)$ for all $u \in T$, which means $$d_H\leq d_{S,T}\leq \frac{2(l-1)|T|}{2|S||T|}=\frac{l-1}{l}\leq \frac{m-1}{m}$$ and we get a contradiction.

%	\item
{\bf $H$ contains only even vectors:} Since the $v_i$'s are even the dimension of $\spun\{S\}$ is at most $m-1$ (since for example $(1,0,...,0)\notin \spun\{S\}$) thus $l\leq m-1.$ Let $H^*$ be the induced subgraph of $D_{m+1}$ with vertices $V(H^*)=\{(1,v_i)|v_i\in V(H))\}$. It is easy to see that all the vectors of $H^*$ are odd, $((1,v_i),(1,v_j))\in E(H^*)\text{ if and only if } (v_i,v_j)\in E(H)$, and the dimension of $\spun\{V(H^*)\}$ is at most $l+1 \le m.$ Having already proven the case for odd vectors, we conclude that
\begin{align*}
d_H=d_{H^*}&\le \frac{\dim(\spun\{V(H^*)\})-1}{\dim(\spun\{V(H^*)\})}\\
&\leq\frac{l+1-1}{l+1}\\
&\leq\frac{m-1}{m},
\end{align*}
and we get a contradiction.	
\end{IEEEproof}

%%%%%%%%%%%%%%%%%%%%%%%%%%%%%%%%%%%%%%%%%%%%%%%%%%%%%%%%%%%%
%
%
\section{Finite Field Size of a Code}\label{section 5}
%
%
%%%%%%%%%%%%%%%%%%%%%%%%%%%%%%%%%%%%%%%%%%%%%%%%%%%%%%%%%%%%
In this section, we address the problem of finding proper coefficients in the parities in order to make the code MDS.
We have already shown that if a code is over some large enough finite field $\mathbb{F}$, it can be made MDS (Theorem \ref{zigzag-sets}).
And we have shown that the optimal code in Theorem \ref{orthogonal-permutations} needs only field of size $3$.
In the following, we will discuss in more details on the field size required to make two kinds of codes MDS: (1) duplication of the optimal code in Corollary \ref{thm2}, and (2) a modification of the code in Example \ref{xmpl1}. Note that both the codes have asymptotic optimal ratio.

Consider the duplication of the optimal code (the code in Corollary \ref{thm2}).
For the $s$-duplication code $\mathcal{C}'$ in Theorem \ref{thm2}, denote the coefficients for the element in row $i$ and column $j^{(t)}$ by $\alpha_{i,j}^{(t)}$ and $\beta_{i,j}^{(t)}$, $0 \le t \le s-1$. Let $\mathbb{F}_q$ be a field of size $q$, and suppose its elements are $\{0,a^0,a^1,\dots,a^{q-2}\}$ for some primitive element $a$.

\begin{cnstr} \label{cons4}
For the $s$-duplication code $\mathcal{C}'$ in Theorem \ref{thm2} over $\mathbb{F}_q$, assign $\alpha_{i,j}^{(t)}=1$ for all $i,j,t$.
For odd $q$, let $s \le q-1$ and assign for all $t \in [0,s-1]$
$$\beta_{i,j}^{(t)}= \left\{
\begin{array}{ll}
a^{t+1}, & \text{if }u_j \cdot i=1 \\
a^{t},  & \text{o.w.}
\end{array}
\right.
$$
where $u_j=\sum_{l=0}^{j}e_l$. For even $q$ (power of 2), let $s \le q-2$ and assign for all $t \in [0,s-1]$
$$\beta_{i,j}^{(t)}= \left\{
\begin{array}{ll}
a^{-t-1}, & \text{if }u_j \cdot i=1 \\
a^{t+1},  & \text{o.w.}
\end{array}
\right.
$$
\end{cnstr}

Notice that the coefficients in each duplication has the same pattern as Construction \ref{cons3} except that values 1 and 2 are replaced by $a^t$ and $a^{t+1}$ if $q$ is odd (or $a^{t+1}$ and $a^{-t-1}$ if $q$ is even).

\begin{thm} \label{thm3}
Construction \ref{cons4} is an $(s(m+1)+2,s(m+1))$ MDS code.
\end{thm}

\begin{IEEEproof}
For the two elements in columns $i^{(t_1)},i^{(t_2)}$ and row $r$, $t_1 \neq t_2$, we can see that they are in the same row set and the same zigzag set. The corresponding two equations from the two parities are linearly independent iff
\begin{equation} \label{0519}
\beta_{r,i}^{(t_1)} \neq \beta_{r,i}^{(t_2)},
\end{equation}
which is satisfied by the construction.

For the four elements in columns $i^{(t_1)},j^{(t_2)}$ and rows $r,r'=r+e_i+e_j$, $0 \le t_1, t_2 \le s-1$, $0 \le i < j \le m$, the code is MDS if
$$\beta_{r,i}^{(t_1)} \beta_{r',i}^{(t_1)} \neq \beta_{r,j}^{(t_2)} \beta_{r',j}^{(t_2)}$$
by \eqref{eq1}. By the remark after Theorem \ref{thm0506}, we know that  $\beta_{r,i}^{(t_1)}\neq \beta_{r',i}^{(t_1)}$, and $\beta_{r,j}^{(t_2)}= \beta_{r',j}^{(t_2)}=a^x$ for some $x$. When $q$ is odd,
$$\beta_{r,i}^{(t_1)} \beta_{r',i}^{(t_1)}=a^{t_1}a^{t_1+1}=a^{2t_1+1} \neq a^{2x}$$
for any $x$ and $t_1$. When $q$ is even,
$$\beta_{r,i}^{(t_1)} \beta_{r',i}^{(t_1)}=a^{t_1+1}a^{-t_1-1}=a^0 \neq a^{2x}$$ for any $t_1$ and $1 \le x \le q-2$ (mod $q-1$). And by construction, $x=t_2+1$ or $x=-t_2-1$ for $0 \le t_2 \le s-1 \le q-3$, so $1 \le x \le q-2$ (mod $q-1$).
Hence, the construction is MDS.
\end{IEEEproof}

{\bf Remark:} For two identical permutations $f_i^{(t_1)}=f_i^{(t_2)}$, \eqref{0519} is necessary and sufficient condition for an MDS code.

\begin{thm}
For an MDS $s$-duplication code, we need a finite field $\mathbb{F}_q$ of size  $q \ge s+1$. Therefore, Theorem \ref{thm3} is optimal for odd $q$.
\end{thm}
\begin{IEEEproof}
Consider the two information elements in row $i$ and columns $j^{(t_1)},j^{(t_2)}$, which are in the same row and zigzag sets, for $t_1 \neq t_2 \in [0,s-1]$. The code is MDS only if
$$\left[ \begin{array}{cc}
\alpha_{i,j}^{(t_1)} & \alpha_{i,j}^{(t_2)} \\
\beta_{i,j}^{(t_1)} & \beta_{i,j}^{(t_2)}
\end{array} \right]$$
has full rank. All the coefficients are nonzero (consider erasing a parity column and a systematic column). Thus,
$(\alpha_{i,j}^{(t_1)})^{-1} \beta_{i,j}^{(t_1)} \neq (\alpha_{i,j}^{(t_2)})^{-1} \beta_{i,j}^{(t_2)}$, and $(\alpha_{i,j}^{(t)})^{-1} \beta_{i,j}^{(t)}$ are distinct nonzero elements in $\mathbb{F}_q$, for $t \in [0,s-1]$. So $q \ge s+1$.
\end{IEEEproof}

For instance, the coefficients in Figure \ref{fig:duplication} are assigned as Construction \ref{cons4} and $\mathbb{F}_3$ is used. One can check that any two column erasures can be rebuilt in this code.

Consider for example an $s$-duplication of the code in Theorem \ref{orthogonal-permutations} with $m=10$, the array is of size $1024 \times (11s+2)$. For $s=2$ and $s=6$, the ratio is $0.522$ and $0.537$ by Corollary \ref{thm2}, the code length is $24$ and $68$, and the field size needed can be $4$ and $8$ by Theorem \ref{thm3}, respectively. Both of these two sets of parameters are suitable for practical applications.

As noted before the optimal construction yields a ratio of $1/2+1/m$ by using duplication of the code in Theorem \ref{orthogonal-permutations}. However the field size is a linear function of the number of duplications of the code.
Is it possible to extend the number of columns in the code while using a constant field size? We know how to get $O(m^3)$ columns by using $O(m^2)$ duplications of the optimal code, however, the field size is $O(m^2)$. The following code construction has roughly the same parameters: $O(m^3)$ columns and an ratio of $\frac{1}{2}+O(\frac{1}{m})$, however it requires only a constant field size of $9$. Actually this construction is a modification of Example \ref{xmpl1}.

\begin{cnstr} \label{k/3}
Let $3|m$, and consider the following set of vectors $S\subseteq \mathbb{F}_2^m$: for each vector $v=(v_1,\dots,v_m) \in S$, $\|v\|_1=3$ and $v_{i_1},v_{i_2},v_{i_3}=1$ for some $i_1 \in [1,m/3],i_2 \in [m/3+1,2m/3], i_3 \in [2m/3+1,m]$. For simplicity, we write $v=\{i_1,i_2,i_3\}$. Construct the $(k+2,k)$ code as in Construction \ref{cnstr1} using the set of vectors $S$, hence the number of systematic columns is $k=|S|=(\frac{m}{3})^3=\frac{m^3}{27}$.
For any $i\in [jm/3+1,(j+1)m/3]$ and some $j=0,1,2$ , define a row vector $M_i = \sum_{l=jm/3+1}^{i}e_l$.
Then define a $m \times 3$ matrix
$$M_v=\left[
\begin{array}{lll}
M_{i_1}^T & M_{i_2}^T & M_{i_3}^T
\end{array}
\right]$$
 for $v=\{i_1,i_2,i_3\}$.
Let  $a$ be a primitive element of  $ \mathbb{F}_9$. Assign the row coefficients as $1$ and the zigzag coefficient for row $r$, column $v$ as $a^t$, where $t=rM_v \in \mathbb{F}_2^3$ (in its binary expansion).
\end{cnstr}

For example, let $m=6$, and $v=\{1,4,6\}=(1,0,0,1,0,1) \in S$. The corresponding matrix is
$$M_v = \left[
\begin{array}{llllll}
1 & 0 & 0 & 0 & 0 & 0 \\
0 & 0 & 1 & 1 & 0 & 0 \\
0 & 0 & 0 & 0 & 1 & 1
\end{array}
\right]^T.
$$
For row $r=26=(0,1,1,0,1,0)$, we have
$$ t = rM_v  = (0,1,1)=3,$$
and the zigzag coefficient is $a^3$.

\begin{thm} \label{thm:k/3}
Construction \ref{k/3} is a $(k+2,k)$ MDS code with array size $2^m \times (k+2)$ and $k=m^3/27$. Moreover, the rebuilding ratio is $\frac{1}{2}+\frac{9}{2m}$ for large $m$.
\end{thm}
\begin{IEEEproof}
For each vector $v \in S$, there are $3(m/3-1)^2$ vectors $u\in S$ such that they have one $1$ in the same location as $v$, i.e. $|v\backslash u|=2$. Hence by Theorem \ref{th:123} and Lemma \ref{lemma 3}, for large $m$ the ratio is
$$\frac{1}{2}+\frac{3((\frac{m}{3})-1)^2}{2(\frac{m^3}{27}+1)} \approx \frac{1}{2} + \frac{9}{2m}.$$

Next we show that the MDS property of the code holds. Consider columns $u,v$ for some $u=\{i_1,i_2,i_3\} \neq v=\{j_1,j_2,j_3\}$ and $i_1, j_1 \in [1,m/3],i_2, j_2 \in [m/3+1,2m/3], i_3,j_3 \in [2m/3+1,m]$. Consider rows $r$ and  $r'=r+u+v$. The condition for the MDS property from \eqref{eq1} becomes
\begin{equation}
a^{rM_{u}^T +  r'M_{u}^T\mod 8}\neq a^{rM_{v}^T +  r'M_{v}^T\mod 8}
\label{eq:eq123}
\end{equation}
where each vector of length $3$ is viewed as an integer in $[0,7]$ and the addition is usual addition mod 8.
Since $v\neq u$, let $l\in[1,3]$ be the largest index such that $i_l\neq j_l$. W.l.o.g. assume that $i_l<j_l$, hence by the remark after Theorem \ref{thm0506}
\begin{equation} \label{eq17}
rM_{i_l}^T \neq  r'M_{i_l}^T
\end{equation}
 and
\begin{equation} \label{eq18}
rM_{j_l}^T = r'M_{j_l}^T.
\end{equation}
Note that for all $t$, $l<t\le 3$, $i_t=j_t$, then since $r'M_{i_t}^T=(r+e_{i_t}+e_{j_t})M_{i_t}^T=rM_{i_t}^T$, we have
\begin{equation}\label{eq05012}
rM_{i_t}^T=r'M_{i_t}^T=rM_{j_t}^T=r'M_{j_t}^T.
\end{equation}
It is easy to infer from \eqref{eq17},\eqref{eq18},\eqref{eq05012} that the $l$-th bit in the binary expansions of $rM_{u}^T +  r'M_{u}^T\mod 8$ and $rM_{v}^T +  r'M_{v}^T\mod 8$ don't equal. Hence \eqref{eq:eq123} is satisfied, and the result follows.
\end{IEEEproof}

Notice that if we do mod $15$ in \eqref{eq:eq123} instead of mod $8$, the proof still follows because $15$ is greater than the largest possible sum in the equation. Therefore, a field of size $16$ is also sufficient to construct an MDS code, and it is easier to implement in a storage system.

Construction \ref{k/3} can be easily generalized to any constant $c$ such that it contains $O(m^c)$ columns and it uses the field of size at least $2^c+1$. For simplicity assume that $c|m$, and simply construct the code using the set of vectors $\{v\}\subset \mathbb{F}_2^m$ such that $\|v\|_1=c$, and for any $j \in [0,c-1]$, there is unique $i_j \in [jm/c+1,(j+1)m/c]$ and $v_{i_j}=1$. Moreover, the finite field of size $2^{c+1}$ is also sufficient to make it an MDS code.
When $c$ is odd the code has ratio of $$\frac{1}{2}+\frac{c^2}{2m}$$ for large $m$.

%%%%%%%%%%%%%%%%%%%%%%%%%%%%%%%%%%%%%%%%%%%%%%%%%%%
%
%
\section{Decoding of the Codes} \label{sec:dec}
%
%
%%%%%%%%%%%%%%%%%%%%%%%%%%%%%%%%%%%%%%%%%%%%%%%%%%%%
In this section, we will discuss decoding algorithms of the proposed codes in case of column erasures as well as a column error. The algorithms work for both Construction \ref{cnstr1} and its duplication code.

Let $\cC$ be a $(k+2,k)$ MDS array code defined by Construction \ref{cnstr1} (and possibly duplication). The code has array size $2^m \times (k+2)$. Let the zigzag permutations be $f_j$, $j \in [0,k-1]$, which are not necessarily distinct.
Let the information elements be $a_{i,j}$, and the row and zigzag parity elements be $r_{i}$ and $z_{i}$, respectively, for $i\in [0,2^m-1],j \in [0,k-1]$.
W.l.o.g. assume the row coefficients are $\alpha_{i,j}=1$ for all $i,j$. And let the zigzag coefficients be $\beta_{i,j}$ in some finite field $\mathbb{F}$.

The following is a summary of the erasure decoding algorithms mentioned in the previous sections.
\begin{alg}
\label{alg0506}(Erasure Decoding)\\
{\bf One erasure.} \\
1) One parity node is erased. Rebuild the row parity by
\begin{equation} \label{eq10}
r_{i}=\sum_{j=0}^{k-1}a_{i,j},
\end{equation}
and the zigzag parity by
\begin{equation} \label{eq11}
z_{i}=\sum_{j=0}^{k-1}\beta_{f_j^{-1}(i),j} a_{f_j^{-1}(i),j}.
\end{equation}
2) One information node $j$ is erased. Rebuild the elements in rows $X_j$ (see Construction \ref{cnstr1}) by rows, and those in rows $\overline{X_j}$ by zigzags. \\
{\bf Two erasures.} \\
1) Two parity nodes are erased. Rebuild by \eqref{eq10} and \eqref{eq11}.\\
2) One parity node and one information node is erased. If the row parity node is erased, rebuild by zigzags; otherwise rebuild by rows.\\
3) Two information nodes $j_1$ and $j_2$ are erased. \\
- If $f_{j_1}=f_{j_2}$, for any $i \in [0,2^m-1]$, compute
\begin{equation} \label{eq12}
\begin{array}{lll}
x_i&=&r_{i}-\sum_{j \neq j_1,j_2} a_{i,j}  \\
y_i&=&z_{f_{j_1}(i)} - \sum_{j \neq j_1,j_2}\beta_{f_j^{-1}f_{j_1}(i),j} a_{f_j^{-1}f_{j_1}(i),j}.
\end{array}
\end{equation}
Solve $a_{i,j_1}, a_{i,j_2}$ from the equations
$$\left[
\begin{array}{ll}
	1 & 1 \\
	\beta_{i,j_1} & \beta_{i,j_2}
\end{array}
\right]
\left[
\begin{array}{l}
a_{i,j_1} \\
a_{i,j_2}
\end{array}
\right]
= \left[
\begin{array}{l}
x_i \\
y_i
\end{array}
\right].
$$
- Else, for any $i \in [0,2^m-1]$, set $i'=i+f_{j_1}(0)+f_{j_2}(0)$, and compute $x_i,x_{i'},y_i,y_{i'}$ according to \eqref{eq12}. Then solve $a_{i,j_1},a_{i,j_2},a_{i',j_1},a_{i',j_2}$ from equations
$$
\left[\begin{array}{cccc}
1 & 1 & 0 & 0 \\
0 & 0 & 1 & 1 \\
\beta_{i,j_1} & 0 & 0 & \beta_{i',j_2} \\
0 & \beta_{i,j_2} & \beta_{i',j_1} & 0
\end{array} \right]
\left[ \begin{array}{c}
a_{i,j_1} \\
a_{i,j_2} \\
a_{i',j_1} \\
a_{i',j_2} \\
\end{array} \right]
= \left[ \begin{array}{c}
x_i \\
x_{i'} \\
y_i \\
y_{i'}
\end{array} \right].
$$
\end{alg}

In case of a column error, we first compute the syndrome, then locate the error position, and at last correct the error.
Let $x_0,x_1,\dots,x_{p-1} \in \mathbb{F}$. Denote $f^{-1}(x_0,x_1,\dots,x_{p-1})=(x_{f^{-1}(0)},x_{f^{-1}(1)},\dots,x_{f^{-1}(p-1)})$ for a permutation $f$ on $[0,p-1]$.
The detailed algorithm is as follows.

\begin{alg} \label{alg2} (Error Decoding) \\
 Compute for all $i \in [0,2^m-1]$:
\begin{eqnarray*}
s_{i,0}&=&\sum_{j=0}^{k-1}a_{i,j}-r_{i} \\
s_{i,1}&=&\sum_{j=0}^{k-1}\beta_{f_j^{-1}(i),j} a_{f_j^{-1}(i),j} -z_i.
\end{eqnarray*}
Let the syndrome be $S_{0}=(s_{0,0},s_{1,0},\dots,s_{2^m-1,0})$ and $S_{1}=(s_{0,1},s_{1,1},\dots,s_{2^m-1,1})$. \\
- If $S_0=0$ and $S_1=0$, there is no error. \\
- Else if one of $S_0, S_1$ is $0$, there is an error in the parity. Correct it by \eqref{eq10} or \eqref{eq11}.\\
- Else, find the error location. For $j = 0$ to $k-1$: \\
\tab Compute for all $i \in [0,2^m-1]$, $x_{i,j} = \beta_{i,j}s_{i,0}.$ \\
\tab Let $X_j=(x_{0,j},\dots,x_{2^m-1,j})$ and $Y_j=f_j^{-1}(X_j)$. \\
\tab If $Y_j=S_1$, subtract $S_0$ from column $j$. Stop.\\
If no such $j$ is found, there are more than one error.
\end{alg}

If there is only one error, the above algorithm is guaranteed to find the error location and correct it, since the code is MDS, as the following theorem states.

\begin{thm}
Algorithm \ref{alg2} can correct one column error.
\end{thm}
\begin{IEEEproof}
Notice that each zigzag permutation $f_j$ is the inverse of itself by Construction \ref{cnstr1}, or $f_j=f_j^{-1}$.
Suppose there is error in column $j$, and the error is $E=(e_{0},e_{1},\dots,e_{2^m-1})$. So the received column $j$ is the sum of the original information and $E$. Thus the syndromes are
$s_{i,0}=e_{i}$ and
$$s_{i,1}=\beta_{f_j(i),j}e_{f_j(i)}.$$
For column $t$, $t \in [0,k-1]$, we have $x_{i,t}=\beta_{i,t}s_{i,0}=\beta_{i,t}e_{i}$. Write $Y_t=f_j^{-1}(X_j)=(y_{0,t},\dots,y_{2^{m}-1,t})$ and then
$$y_{i,t}=x_{f_t(i),t}=\beta_{f_t(i),t}e_{f_t(i)}.$$
We will show the algorithm finds $Y_t=S_1$ iff $t=j$, and therefore subtracting $S_0=E$ from column $j$ will correct the error. When $t=j$,
$y_{i,t}=s_{i,1},$
for all $i \in [0,2^m-1]$, so $Y_j=S_1$. Now suppose there is $t \neq j$ such that $Y_t=S_1$. Since the error $E$ is nonzero, there exists $i$ such that $e_{f_j(i)}\neq 0$. Consider the indices $i$ and $i'=f_t f_j(i)$. $y_{i,t}=s_{i,1}$ yields
\begin{equation}\label{05191}
\beta_{f_t(i),t}e_{f_t(i)}=\beta_{f_j(i),j}e_{f_j(i)}.
\end{equation}
{\bf Case 1}: When $f_t=f_j$, set $r=f_t(i)=f_j(i)$, then \eqref{05191} becomes $\beta_{r,t}e_r=\beta_{r,j}e_r$ with $e_r \neq 0$. Hence $\beta_{r,t}=\beta_{r,j}$ which contradicts \eqref{0519}. \\
{\bf Case 2}: When $f_t \neq f_j$, since $f_t,f_j$ are commutative and are inverse of themselves, $f_t(i')=f_t f_t f_j(i)=f_j(i)$ and $f_j(i')=f_j f_t f_j(i)=f_t(i)$. Therefore $y_{i',t}=s_{i',1}$ yields
\begin{equation} \label{eq1011}
\beta_{f_j(i),t}e_{f_j(i)}=\beta_{f_t(i),j}e_{f_t(i)}.
\end{equation}
The two equations \eqref{05191} \eqref{eq1011} have nonzero solution $(e_{f_j(i)},e_{f_t(i)})$ iff
$$\beta_{f_t(i),t}\beta_{f_j(i),t} = \beta_{f_j(i),j}\beta_{f_t(i),j},$$
which contradicts \eqref{eq1} with $r=f_t(i),r'=f_j(i)$. Hence the algorithm finds the unique erroneous column.
\end{IEEEproof}

If the computations are done in parallel for all $i \in [0,2^m-1]$, then  Algorithm \ref{alg2} can be done in time $O(k)$. Moreover, since the permutations $f_i$'s only change one bit of a number in $[0,2^m-1]$ in the optimal code in Theorem \ref{orthogonal-permutations}, the algorithm can be easily implemented.

%%%%%%%%%%%%%%%%%%%%%%%%%%%%%%%%%%%%%%%%%%%%%%%%%%%%%%%%%%%

\section{Generalization of the Code Construction}
\label{generalization}

%%%%%%%%%%%%%%%%%%%%%%%%%%%%%%%%%%%%%%%%%%%%%%%%%%%%%%%%%%%%
In this section we generalize Construction \ref{cnstr1} to arbitrary number of parity nodes. Let $n-k=r$ be the number of parity nodes. We will construct an $(n,k)$ MDS array code, i.e., it can recover from up to $r$ node erasures for arbitrary integers $n,k$.
We will show this code has optimal rebuilding ratio of $1/r$ when a systematic node is erased.
We assume that each systematic nodes stores $\frac{\cM}{k}$ of the information and corresponds to columns $[0,k-1]$. The $i$-th parity node is stored in column $k+i$, $0\leq i\leq r-1$, and is associated with zigzag sets $\{Z_{j}^i: j \in [0,p-1]\}$, where $p$ is the number of rows in the array.

\begin{cnstr}
\label{cnstr5}
Let the information array be $A=(a_{i,j})$ with size $r^m \times k$ for some integers $k,m$.
Let $T=\{v_0,...,v_{k-1}\}\subseteq \mathbb{Z}_r^m$ be a subset of vectors of size $k$, where for each $v=(v_1,...,v_m)\in T,$
\begin{equation}
\gcd(v_1,...,v_m,r)=1,
\label{eq:45678}
\end{equation}
where $\gcd$ is the greatest common divisor. For any $l$, $0\leq l\leq r-1$, and $v\in T$ we define the permutation $f^l_v:[0,r^m-1]\rightarrow [0,r^m-1]$ by $f_v^l(x)=x+lv$, where by abuse of notation we use $x\in[0,r^m-1]$ both to represent the integer and its $r$-ary representation, and all the calculations are done over $\mathbb{Z}_r$.For example, for $m=2,r=3, x=4, l=2, v=(0,1)$,
$$f_{(0,1)}^2(4)=4+2(0,1)=(1,1)+(0,2)=(1,0)=3.$$
One can check that the permutation $f_{(0,1)}^2$ in a vector notation is $[2,0,1,5,3,4,8,6,7]$.
For simplicity denote the permutation $f^l_{v_j}$ as $f^l_j$ for $v_j\in T$.
For $t \in [0,r^m-1]$, we define the zigzag set $Z_t^l$ in parity node $l$ as the elements $a_{i,j}$ such that their coordinates  satisfy
$f^l_j(i)=t.$
In a rebuilding of systematic node $i$ the elements in rows $X^l_i=\{x\in [0,r^m-1]:x\cdot v_i=r-l\}$ are rebuilt by parity node $l$, $l \in [0,r-1]$.
From \eqref{eq:45678} we get that for any $i$ and $l$,
$|X^l_i|=r^{m-1}.$
\end{cnstr}

Note that similar to Theorem \ref{zigzag-sets}, using a large enough field, the parity nodes described above form an $(n,k)$ MDS array code under appropriate selection of coefficients in the linear combinations of the zigzags.

Consider the rebuilding of systematic node $i\in[0,k-1]$. In a systematic column $j\neq i$ we need to access all the elements that are contained in the sets that belong to the rebuilding set of column $i$. Namely, in column $j$ we need to access the elements in rows
\begin{equation}
\cup_{l=0}^{r-1} f_j^{-l} f_i^{l} (X_i^l).
\label{eq:1212}
\end{equation}
\eqref{eq:1212} follows since the zigzags $Z_t^l$ for any $t\in f_i^l(X_i^l)$ are used to rebuild the elements of column $i$ in rows $X_i^l$. Moreover the element in column $j$ and zigzag $Z_t^l$ is $a_{f^{-l}_j(t),j}$.
The following lemma will help us to calculate the size of \eqref{eq:1212}, and in particular calculating the ratio of codes constructed by Construction \ref{cnstr5}.

\begin{lem} \label{lem:orth}
For any $v=(v_1,...v_m),u \in \mathbb{Z}_r^m$ and $l\in [0,r-1]$ such that $\gcd(v_1,...,v_m,r)=1$, define $c_{v,u}=v\cdot(v-u)-1$. Then
\begin{displaymath}
   |f_u^{-i} f_v^{i}(X_v^i) \cap f_u^{-j} f_v^{j}(X_v^j)|
=  \left\{
\begin{array}{ll}
|X_v^0|, & (i-j)c_{v,u}=0 \\
0,       & \text{o.w.}
\end{array}
\right.
\end{displaymath}
In particular for $j=0$ we get
\begin{displaymath}
|f_u^{-l}f_v^l (X_v^l) \cap X_v^0| = \left\{
\begin{array}{ll}
|X_v^0|, & \text{if } lc_{v,u}=0 \\
0,       & \text{o.w.}
%\label{eq:0010}
\end{array}
\right.
\end{displaymath}

 \end{lem}

\begin{IEEEproof}
Consider the group $(\mathbb{Z}_r^m,+)$. Note that $X_v^0=\{x: x\cdot v=0\}$ is a subgroup of $\mathbb{Z}_r^m$ and $X_v^i=\{x: x\cdot v=r-i\}$ is its coset. Therefore, $X_v^i=X_v^0 + a_v^i, X_v^j=X_v^0+a_v^j$, for some $a_v^i \in X_v^i, a_v^j \in X_v^j$. Hence $f_u^{-i} f_v^{i}(X_v^i)=X_v^0+a_v^i+i(v-u)$ and $f_u^{-j} f_v^{j}(X_v^j)=X_v^0+a_v^j+j(v-u)$ are cosets of $X_v^0$. So they are either identical or disjoint. Moreover they are identical if and only if
$$a_v^i-a_v^j+(i-j)(v-u)\in X_v^0,$$
i.e., $(a_v^i-a_v^j+(i-j)(v-u))\cdot v =0$. But by definition of $X_v^i$ and $X_v^j$, $a_v^i \cdot v=-i$, $a_v^j \cdot v=-j$, so $(i-j)\cdot c_{v,u}=0$ and the result follows.
\end{IEEEproof}
The following theorem gives the ratio for any code of Construction \ref{cnstr5}.
\begin{thm} \label{thm: gen rate}
The ratio for the code constructed by Construction \ref{cnstr5} and set of vectors $T$  is
$$\frac{\sum_{v \in T} \sum_{u \neq v \in T} \frac{1}{\gcd(r,c_{v,u})} + |T| }{|T|(|T|-1+r)},$$
which also equal to
$$\frac{1}{r}+\frac{\sum_{v\in T}\sum_{u\in T,u\neq v}|F_{u,v}(\overline{X_v^0})\cap \overline{X_v^0}|}{|T|(|T|-1+r)r^m}.$$
Here we define the function $F_{u,v}(t)=f_u^{-i}f_v^i(t)$ for $t\in X_v^i.$
\end{thm}
\begin{IEEEproof}
By \eqref{eq:1212} and noticing that we access $r^{m-1}$ elements in each parity node, the ratio is
\begin{equation} \label{eq:union}
\frac{\sum_{v \in T} \sum_{u \neq v \in T} |\cup_{i=0}^{r-1} f_u^{-i} f_v^{i} (X_v^i)| + |T|r^m} {|T|(|T|-1+r) r^m}.
\end{equation}
From Lemma \ref{lem:orth}, and noticing that $|\{i: ic_{v,u}=0 \mod r\}|=\gcd(r,c_{v,u})$, we get
$$|\cup_{i=0}^{r-1} f_u^{-i} f_v^{i} (X_v^i)|=r^{m-1}\times r /\gcd(r,c_{v,u}).$$
And the first part follows. For the second part,
\begin{eqnarray}
&&\frac{\sum_{v \in T} \sum_{u \neq v \in T} |\cup_{i=0}^{r-1} f_u^{-i} f_v^{i} (X_v^i)| + |T|r^m} {|T|(|T|-1+r) r^m}\nonumber\\
&=&\frac{\sum_{v \in T} \sum_{u \neq v \in T} |X_v^0|+|\cup_{i=1}^{r-1} f_u^{-i} f_v^{i} (X_v^i)\backslash X_v^0| + |T|r^m} {|T|(|T|-1+r) r^m}\nonumber\\
&=&\frac{1}{r}+\frac{\sum_{v \in T} \sum_{u \neq v \in T} |\cup_{i=1}^{r-1} f_u^{-i} f_v^{i} (X_v^i)\cap \overline{X_v^0}|} {|T|(|T|-1+r) r^m}\nonumber\\
&=&\frac{1}{r}+\frac{\sum_{v\in T}\sum_{u\in T,u\neq v}|F_{u,v}(\overline{X_v^0})\cap \overline{X_v^0}|}{|T|(|T|-1+r)r^m}\label{eq:000}.
\end{eqnarray}
The proof is completed.
\end{IEEEproof}

Notice that $\overline{X_v^0}$ represents elements not accessed for parity 0 (row parity), and $F_{u,v}(\overline{X_v^0})$ are elements accessed for parity $1,2,\dots,r-1$. Therefore $F_{u,v}(\overline{X_v^0})\cap \overline{X_v^0}$ are the elements accessed excluding those for the row parity.
In order to get a low rebuilding ratio, we need to minimize the second term in \eqref{eq:000}. We say that a family of permutation set
$\{ \{f_0^l\}_{l=0}^{r-1},...,\{f_{k-1}^{l}\}_{l=0}^{r-1}\}$ together with sets
$\{\{X_0^l\}_{l=0}^{r-1},...,\{X_{k-1}^l\}_{l=0}^{r-1}\}$ is a family of \emph{orthogonal permutations} if for any $i,j\in [0,k-1]$ the set $\{X_i^l\}_{i=0}^{r-1}$ is a equally sized partition of $[0,r^m-1]$ and
\begin{equation*}
\frac{|F_{j,i}(\overline{X_i^0})\cap \overline{X_i^0}|}{r^{m-1}(r-1)}=\delta_{i,j}.
%\label{eq:0011}
\end{equation*}
One can check that for $r=2$ the definition coincides with the previous definition of orthogonal permutations for two parities.
It can be shown that the above definition is equivalent to that for any $0 \le i \neq j \le k-1,0 \le l \le r-1$,
\begin{equation}
f_j^l(X_i^0) = f_i^l (X_i^l).
\label{eq:1122}
\end{equation}
For a set of orthogonal permutations, rebuilding ratio is $1/r$ by  \eqref{eq:000}, which is optimal according to (\ref{eq:tradeoff}),

Now we are ready to construct a code with optimal rebuilding ratio and $r$ parities.
\begin{thm}
\label{mashumashu}
The set $\{ \{f_0^l\}_{l=0}^{r-1},...,\{f_{m}^{l}\}_{l=0}^{r-1}\}$together with set
$\{\{X_0^l\}_{l=0}^{r-1},...,\{X_{m}^l\}_{l=0}^{r-1}\}$ constructed by the vectors $\{e_i\}_{i=0}^m$ and Construction \ref{cnstr5}, where $X_0^l$ is modified to be $X_0^l=\{x\in \mathbb{Z}_r^m:x\cdot (1,1,...,1)=l\}$ for any $l\in [0,r-1]$
is a family of \emph{orthogonal permutations}. Moreover the corresponding $(m+1+r,m+1)$ code has \emph{optimal} ratio of $\frac{1}{r}.$
\end{thm}
\begin{IEEEproof}
For $1 \le i \neq j \le m$, $c_{i,j}=e_i \cdot(e_i-e_j)-1 = 0$, hence by Lemma \ref{lem:orth} for any $l\in[0,r-1]$
$$f_j^{-l}f_i^{l}(X_i^l)\cap X_i^0=X_i^0,$$ and \eqref{eq:1122} is satisfied.
For $1 \le i \le m$, and all $0 \le l \le r-1$,
\begin{eqnarray*}
f_0^{-l}f_i^l(X_i^l) &=& f_i^l(\{v:v_i=-l\})=\{v+le_i:v_i=-l\} \\
&=&\{v:v_i=0\}=X_i^0
\end{eqnarray*}
Therefore,
$f_0^{-l}f_i^{l}(X_i^l)\cap X_i^0=X_i^0$, and \eqref{eq:1122} is satisfied.
Similarly,
\begin{align*}
f_i^{-l}f_0^l(X_0^l)&=f_i^{-l}(\{v:v\cdot (1,...,1)=l\})\\
&=\{v-le_i:v\cdot (1,...,1)=l\}\\
&=\{v:v\cdot (1,...,1)=0\}=X_0^0.
\end{align*}
Hence again \eqref{eq:1122} is satisfied and this is a family of orthogonal permutations, and the result follows.
\end{IEEEproof}

\begin{figure*}
	\centering
		\includegraphics[scale=.6]{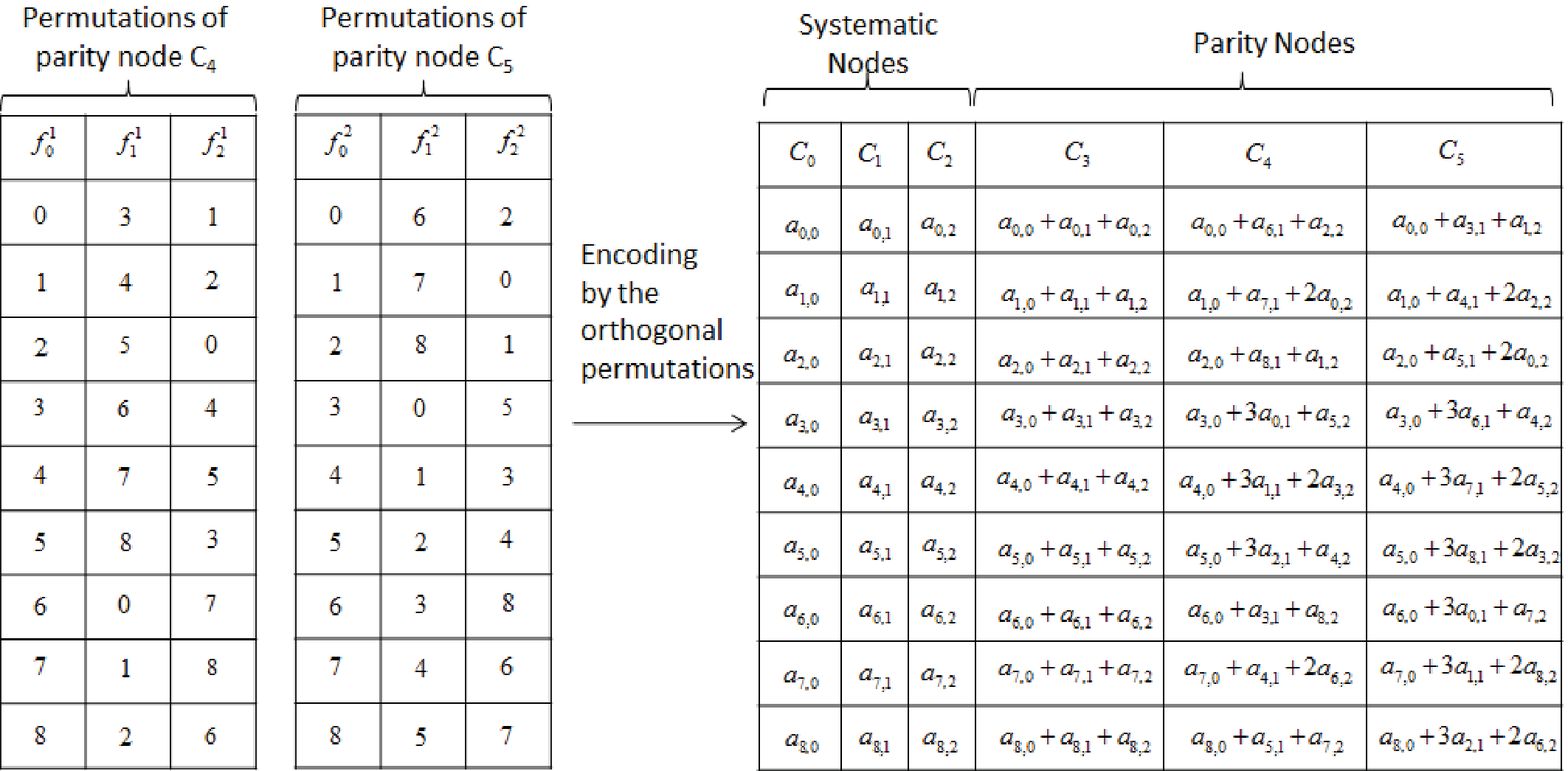}
	\caption{A $(6,3)$ MDS array code with optimal ratio $1/3$. The first parity $C_3$ corresponds to the row sums, and the corresponding identity permutations are omitted. The second and third parity $C4,C_5$ are generated by the permutations $f_i^1,f_i^2$ respectively, $i=0,1,2$. The elements are from $\mathbb{F}_7$, where $3$ is a primitive element.}
	\label{fig5}
\end{figure*}

Surprisingly, one can infer from the above theorem that changing the number of parities from $2$ to $3$ adds only one node to the system, but reduces the ratio from $1/2$ to $1/3$ in the rebuilding of any systematic column.

The example in Figure \ref{fig5} shows a code with $3$ systematic nodes and $3$ parity nodes constructed by Theorem \ref{mashumashu} with $m=2$. The code has an optimal ratio of $1/3$. For instance, if column $C_1$ is erased, accessing rows $\{0,1,2\}$ in the remaining nodes will be sufficient for rebuilding.

Similar to the $2$ parity case, the following theorem shows that Theorem \ref{mashumashu} achieves the optimal number of columns. In other words,
the number of rows has to be exponential in the number of columns in any systematic MDS code with optimal ratio, optimal update, and $r$ parities. This follows since any such optimal code is constructed from a family of orthogonal permutations.

\begin{thm} \label{thm0519}
Let $\{ \{f_0^l\}_{l=0}^{r-1},...,\{f_{k-1}^l\}_{l=0}^{r-1} \}$ be a family of orthogonal permutations over the integers $[0,r^m-1]$ together with the sets $\{ \{X_0^l\}_{l=0}^{r-1},...,\{X_{k-1}^l\}_{l=0}^{r-1} \}$, then $k\leq m+1$.
\end{thm}
\begin{IEEEproof}
We prove it by induction on $m$.
When $m=0$, it is trivial that $k \le 1$.
Now suppose we have a family of orthogonal permutations $\{ \{f_0^l\}_{l=0}^{r-1},...,\{f_{k-1}^l\}_{l=0}^{r-1} \}$ over $[0,r^m-1]$, and we will show $k \le m+1$.
Recall that orthogonality is equivalent \eqref{eq:1122}.
Notice that for any permutations $g,h_0,...,h_{r-1},\{ \{h_lf_0^lg\}_{l=0}^{r-1},...,\{h_lf_{k-1}^lg\}_{l=0}^{r-1} \} \}$ are still a family of orthogonal permutations with sets $\{ \{g^{-1}(X_0^l)\},...,\{g^{-1}(X_{k-1}^l)\} \}$.
This is because
\begin{align*}
h_l f_j^l g (g^{-1}(X_i^0))&= h_l f_j^l(X_i^0)\\
&=h_l f_i^l (X_i^l)\\
&=h_l f_i^l g(g^{-1}(X_i^l)).
\end{align*}
Therefore, w.l.o.g. we can assume  $X_0^{l}=[lr^{m-1},(l+1)r^{m-1}-1]$, and $f_0^{l}$ is the identity permutation, for $0 \le l \le r-1$.

Let $1 \leq i \neq j \le k-1,l\in [0,r-1]$ and define
\begin{eqnarray*}
A&=&f_j^l(X_i^0)=f_i^l(X_i^l), \\
B &=& f_j^l(X_i^0 \cap X_0^0), \\
C &=& f_i^l(X_i^l \cap X_0^0).
\end{eqnarray*}
Therefore $B,C$ are subsets of $A$, and their compliments in $A$ are
\begin{eqnarray*}
&A \backslash B = f_j^l(X_i^0 \cap \overline{X_0^0}), \\
&A \backslash C = f_i^l(X_i^l \cap \overline{X_0^0}).
\end{eqnarray*}
From \eqref{eq:1122} for any $j\neq 0$,
\begin{equation}
f_j^l(X_0^0)=f_0^l(X_0^l)=X_0^l
\label{eq:566}
\end{equation}hence,
\begin{equation}
B,C \subseteq X_0^l
\label{eq:asdf}
\end{equation}
Similarly, for any $j\neq 0$,
$ f_j^l(\overline{X_0^0})= \overline{f_j^l(X_0^0)} =\overline{X_0^l}$, hence
\begin{equation}
A \backslash B,A \backslash C \subseteq \overline{X_0^l}.
\label{eq:asdf2}
\end{equation}
From \eqref{eq:asdf},\eqref{eq:asdf2} we conclude that $B=C=A \cap X_0^l$, i.e.,
\begin{equation}
f_j^l(X_i^0 \cap X_0^0)=f_i^l(X_i^l \cap X_0^0).
\label{eq:2340}
\end{equation}
For each $l\in[0,r-1],j\in [1,k-1]$ define $\hat{f}_j^l(x)=f_j^l(x)-lr^{m-1}$ and $\hat{X_j^l}=X_j^l \cap X_0^0$ then,
\begin{align}
\hat{f}_j^l([0,r^{m-1}-1])&=f_j^l(X_0^0)-lr^{m-1}\nonumber\\
&=X_0^l-lr^{m-1}\label{wewe}\\
&=[0,r^{m-1}-1]\nonumber,
\end{align}
where \eqref{wewe} follows from \eqref{eq:566}. Moreover, since $f_i^l$ is bijective we conclude that $\hat{f_i^l}$ is a permutation on $[0,r^{m-1}-1].$
\begin{align}
\hat{f}_i^l(\hat{X_i^l})&=f_i^l(X_i^l\cap X_0^0)-lr^{m-1}\nonumber\\
&=f_j^l(X_i^0\cap X_0^0)-lr^{m-1}\label{eq0101}\\
&=\hat{f_j^l}(\hat{X_i^0})\nonumber,
\end{align}
where \eqref{eq0101} follows from \eqref{eq:2340}.
Since $\{X_i^l\}_{l=0}^{r-1}$ is a partition of  $[0,r^m-1]$, then $\{\hat{X}_i^l\}_{l=0}^{r-1}$ is also a partition of $X_0^0=[0,r^{m-1}-1]$. Moreover, since
$\hat{f}_i^l(\hat{X_i^l})=\hat{f_j^l}(\hat{X_i^0})$ for any $l\in [0,r-1]$, and $\hat{f}_i^l,\hat{f}_j^l$ are bijections, we conclude
$$|\hat{X}_i^l|=|\hat{X}_i^0|$$
for all $l \in [0,r-1]$, i.e., $\{\hat{X}_i^l\}$, $l \in [0,r-1]$, is a equally sized partition of $[0,r^{m-1}-1]$.
Therefore $\{\{\hat{f_1^l}\}_{l=0}^{r-1},...,\{\hat{f_{k-1}^l}\}_{l=0}^{r-1} \}$ together with $\{ \{\hat{X_1^l}\}_{l=0}^{r-1},...,\{\hat{X_{k-1}^l}\}_{l=0}^{r-1}\}$ is a family of orthogonal permutations over integers $[0,r^{m-1}-1]$, hence by induction $k-1\leq m$ and the result follows.
\end{IEEEproof}

After presenting the construction of a code with optimal ratio of $1/r$, we move on to deal with the problem of assigning the proper coefficient in order to satisfy the MDS property. This task turns out to be not easy when the number of parities $r>2.$ The next theorem gives a proper assignment for the code with $r=3$ parities, constructed by the optimal construction given before. This assignment gives an upper bound on the required field size.

\begin{thm}
A field of size at most $2(m+1)$ is sufficient to make the code constructed by Theorem \ref{mashumashu} with $r=3$ parities, a $(m+4,m+1)$ MDS code.
\end{thm}
\begin{proof}
Let $\mathbb{F}_q$ be a field of size $q\geq 2(m+1)$. For any  $l\in [0,m]$ let $A_l=(a_{i,j})$ be the representation of the permutation $f_{e_l}^1$ by a permutation matrix with a slight modification and is defined as follows,
\begin{equation*}
a_{i,j}=
\begin{cases}
\alpha^l & f_{e_l}^1(j)=i \text{ and } j\cdot e_l=0\\
1 & f_{e_l}^1(j)=i \text{ and } j\cdot e_l \neq 0\\
0 & \text{otherwise},
\end{cases}
\end{equation*}
where $\alpha$ is a primitive element of $\mathbb{F}_q$.
Let $W$ be the matrix that create the parities nodes, defined as
$$W=\left[
\begin{array}{llll}
	B_0^0 & B_1^0 &...& B_m^0\\
	B_0^1 & B_1^1 &...& B_m^1 \\
	B_0^2 & B_1^2 &...& B_m^2
\end{array}
\right].
$$
Where $B_l^j=(A_l)^j$ for $l\in [0,m]\text{ and } j\in[0,2]$. It easy to see that indeed block row $i\in [0,2]$ in the block matrix $m$ corresponds to parity $i$. We will show that this coefficient assignment satisfy the MDS property of the code. First
we will show that under this assignment of coefficients the matrices $A_l$ \underline{commute}, i.e. for any $l_1 \neq l_2\in [0,m],A_{l_1}A_{l_2}=A_{l_2}A_{l_1}$. For simplicity, write $f_{e_{l_1}}^1=f_1,f_{e_{l_2}}^1=f_2,A_{l_1}=(a_{i,j}),A_{l_2}=(b_{i,j}),3^m=p$.
For a vector $x=(x_0,...,x_{p-1})$ and $y = x A_{l_1}$,
its $j$-th entry satisfies $y_j=a_{f_1(j),j}x_{f_1(j)}$ for all $j \in [0,p-1]$. And by similar calculation, $z=x A_{l_1} A_{l_2} = y A_{l_2}$ will satisfy
$$z_j = b_{f_2(j),j}y_{f_2(j)}=
b_{f_2(j),j} a_{f_1(f_2(j)),f_2(j)}x_{f_1(f_2(j))}.$$
Similarly, if $w=x A_{l_2} A_{l_1}$, then
$$w_j = a_{f_1(j),j} b_{f_2(f_1(j)),f_1(j)}x_{f_2(f_1(j))}.$$
Notice that $$ f_1(j) \cdot e_{l_2} = (j+e_{l_1})e_{l_2}=j \cdot e_{l_2},$$
so  $b_{f_2(j),j}=b_{f_2(f_1(j)),f_1(j)}$. Similarly, $a_{f_1(j),j}=a_{f_1(f_2(j)),f_2(j)}$. Moreover, $$f_1(f_2(j))=f_2(f_1(j))=j+e_{l_1}+e_{l_2}.$$ Hence, $z_j=w_j$ for all $j$ and
$$x A_{l_1} A_{l_2} =z=w= x A_{l_2} A_{l_1}$$
for all $x \in \mathbb{F}_3^{m}$. Thus $A_{l_1}A_{l_2}=A_{l_2}A_{l_1}$.

Next we show for any $i$, $\underline{A_i^3=\alpha^iI}$. For any vector $x$, Let $y=xA_i^3.$ Then $$y_j=a_{f_i(j),j}a_{f_i^2(j),f_i(j)}a_{f_i^3(j),f_i^2(j)}x_{f_i^3(j)}.$$
However, $f_i^3(j)=j+3e_i=j $ (since the addition is done over $\mathbb{F}_3^m$), and exactly one of $j\cdot e_i,f_i(j) \cdot e_i,f_i^2(j) \cdot e_i$ equals to $0$. Thus $y_j=\alpha^i x_j$ or $xA_i^3=\alpha^i x$ for any $x$. Hence $A_i^3=\alpha^i I$.

The code is MDS if it can recover from loss of any $3$ nodes. With this assignment of coefficients the code is MDS iff any block sub matrices of sizes $1\times1,2\times 2,3\times 3$ of the matrix $M$
are invertible. The case of $1 \times 1$ sub matrix is trivial. Let $0\leq i<j<k\leq m$ we will see that the $3 \times 3$ matrix
$$\left[
\begin{array}{lll}
	I & I & I\\
	A_{i} & A_j & A_k \\
	A_{i}^2 & A_j^2 & A_k^2
\end{array}
\right]
$$
is invertible. By Theorem $1$ in \cite{Determinants} and the fact that  all the blocks in the matrix commute we get that the determinant of this matrix equals to $ \det(A_k-A_j)\cdot\det(A_k-A_i)\cdot\det(A_j-A_i)$. Hence we need to show that for any $i>j$, $\det(A_i-A_j)\neq 0$, which is equivalent to $\det(A_iA_j^{-1}-I)\neq 0.$ Note that for any $i$, $A_i^3=\alpha^iI$. Denote by $A=A_iA_j^{-1}$, hence $A^3=(A_iA_j^{-1})^3=A_i^3A_j^{-3}=\alpha^{i-j}I\neq I$. Therefore
$$0\neq\det(A^3-I)=\det(A-I)\det(A^2+A+I).$$
Therefore $\det(A-I)=\det(A_iA_j^{-1}-I)\neq 0$.

For a submatrix of size $2 \times 2$, we need to check that for $i>j$
$$\det(\left[
\begin{array}{ll}
	I & I \\
	A_j^2 & A_i^2
\end{array}
\right])=\det(A_j^{2})\det(A_i^2A_j^{-2}-I)\neq 0.
$$
Note that $A^6=(A_iA_j^{-1})^6= \alpha^{2(i-j)}I\neq I$ since $0<i-j\leq m< \frac{q-1}{2}.$ Hence
$$0\neq \det(A^6-I)=\det(A^2-I)(A^4+A^2+I),$$ and $\det(A^2-I)=\det(A_i^2A_j^{-2}-I)\neq 0$ which concludes the proof.
\end{proof}

For example, the coefficients of the parities in Figure \ref{fig5} are assigned as the above proof. Since $m=2$, the field of size $7$ is sufficient. The primitive element is chosen to be $3$. One can check that when losing any three columns we can still rebuild them.
%%%%%%%%%%%%%%%%%%%%%%%%%%%%%%%%%%%%%%%%%%%%%%%%%%%%%%%%%%%%%%%%%%%%%%%%%%

\section{Rebuilding Multiple Erasures} \label{sec:multi}

%%%%%%%%%%%%%%%%%%%%%%%%%%%%%%%%%%%%%%%%%%%%%%%%%%%%%%%%%%%%%%%%%%%%%%%%%%
In this section, we discuss the rebuilding of $e$ erasures, $1 \le e \le r$. We will first prove the lower bound for rebuilding ratio and repair bandwidth. Then we show a construction achieving the lower bound for systematic nodes. At last we generalize this construction and Construction \ref{cnstr5}, and propose a rebuilding algorithm using an arbitrary subgroup and its cosets.

In this section, in order to simplify some of the results we will assume that $r$ is a prime and the calculations are done over $\mathbb{F}_r$. Note that all the result can be generalized with minor changes for an arbitrary integer $r$ and the ring $\mathbb{Z}_r$.
\subsection{Lower Bounds}

The next theorem shows that the rebuilding ratio for Construction \ref{cnstr5} is at least $e/r$.

\begin{thm} \label{thm0419}
Let $A$ be an array with $r$ parity nodes constructed by Construction \ref{cnstr5}. In an erasure of $1\leq e \leq r$ systematic nodes, the rebuilding ratio is at least $\frac{e}{r}$.
\end{thm}

\begin{proof}
In order to recover the information in the systematic nodes we need to use at least $er^m$ zigzag sets from the $r^{m+1}$ sets (There are $r$ parity nodes, $r^m$ zigzag sets in each parity). By the pigeonhole principle there is at least one parity node, such that at least $er^{m-1}$ of its zigzag sets are used. Hence each remaining systematic node has to access its elements that are contained in these zigzag sets. Therefore each systematic node accesses at least $er^{m-1}$ of its information out of $r^m$, which is a portion of $\frac{e}{r}.$

Since we use at least $er^m$ zigzag sets, we use at least $er^m$ elements in the $r$ parity nodes, which is again a portion of $\frac{e}{r}$. Hence the overall rebuilding ratio is at least $\frac{e}{r}$.
\end{proof}

In a general code (not necessary MDS, systematic, or optimal update), what is the amount of information needed to transmit in order to rebuild $e$ nodes? Assume that in the system multiple nodes are erased, and we rebuild these nodes \emph{simultaneously} from information in the remaining nodes. It should be noted that this model is a bit different from the distributed repair problem, where the recovery of each node is done separately.
We follow the definitions and notations of \cite{KumarProof}. An \emph{exact-repair reconstructing code} satisfies the following two properties: (i)Reconstruction: any $k$ nodes can rebuild the total information. (ii)Exact repair: if $e$ nodes are erased, they can be recovered exactly by transmitting information from the remaining nodes.

Suppose the total amount of information is $\cM$, and the $n$ nodes are $[n]$. For $e$ erasures, $1\leq e \leq r$, denote by $\alpha, d_e,\beta_e$ the amount of information stored in each node, the number of nodes connected to the erased nodes, and the amount of information transmitted by each of the nodes, respectively. For subsets $A,B \subseteq [n]$, $W_A$ is the amount of information stored in nodes $A$, and $S_A^B$ is the amount of information transmitted from nodes $A$ to nodes $B$ in the rebuilding.

The following results give lower bound of repair bandwidth for $e$ erasures, and the proofs are based on  \cite{KumarProof}.

\begin{lem}
Let $B\subseteq [n]$ be a subset of nodes of size $|e|$, then for an arbitrary set of nodes $A$, $|A|\leq d_e$ such that $B\cap A=\emptyset$,
$$H(W_B|W_A)\leq \min\{|B|\alpha,(d_e-|A|)\beta_e\}.$$
\end{lem}

\begin{proof}
If nodes $B$ are erased, consider the case of connecting to them nodes $A$ and nodes $C$, $|C|=d_e-|A|.$ Then the exact repair condition requires
\begin{align*}
0&=H(W_B|S_A^B,S_C^B)\\
&=H(W_B|S_A^B)-I(W_B,S_C^B|S_A^B)\\
&\geq H(W_B|S_A^B)-H(S_C^B)\\
&\geq H(W_B|S_A^B) -(d-|A|)\beta_e\\
&\geq H(W_B|W_A) -(d-|A|)\beta_e .
\end{align*}
Moreover, it is clear that $H(W_B|W_A)\leq H(W_B)\leq |B|\alpha$ and the result follows.
\end{proof}

\begin{thm} \label{th0503}
Any reconstructing code with file size  $\cM$ must satisfy for any $1\leq e \leq r$
$$\cM\leq s \alpha+ \sum_{i=0}^{\lfloor \frac{k}{e}\rfloor-1} \min\{e \alpha,(d_e-ie-s)\beta_e\}$$
where $s=k \mod e, 0 \le s < e$. Moreover for an MDS code, $\beta_e \geq \frac{e\cM}{k(d-k+e)}.$
\end{thm}

\begin{proof}
The  file can be reconstructed from any set of $k$ nodes, hence
\begin{align*}
\cM &=H(\cW_{[k]})\\
&=H(W_{[s]})+\sum_{i=0}^{\lfloor\frac{k}{e}\rfloor-1}H(\cW_{[ie+s+1,(i+1)e+s]}|\cW_{[ie+s]})\\
&\leq s \alpha+ \sum_{i=0}^{\lfloor \frac{k}{e}\rfloor-1} \min\{e \alpha,(d_e-ie-s)\beta_e\}.
\end{align*}
In an MDS code $\alpha=\frac{\cM}{k}$, hence in order to satisfy the inequality any summand of the form $\min\{e\alpha,(d_e-ie-s)\beta_e\}$ must be at least $e\frac{\cM}{k}$, which occurs if and only if
$(d_e-(\lfloor \frac{k}{e} \rfloor-1 )e-s)\beta_e\geq \frac{e\cM}{k}$. Hence we get $$\beta_e \geq \frac{e\cM}{k(d-k+e)}.$$
And the proof is completed.
\end{proof}

Therefore, the lower bound of the repair bandwidth for an MDS code is $\frac{e\cM}{k(d-k+e)}$, which is the same as the lower bound of the rebuilding ratio in Theorem \ref{thm0419}.

\subsection{Rebuilding Algorithms}

Next we discuss how to rebuild in case of $e$ erasures, $1 \le e \le r$, for an MDS array code with optimal update. Theorem \ref{th0503} gives the lower bound $e/r$ on the rebuilding ratio for $e$ erasures. Is this achievable? Let us first look at an example.

\begin{xmpl} \label{xmpl0503}
Consider the code in Figure \ref{fig5} with $r=3$. When $e=2$ and columns $C_0,C_1$ are erased, we can access rows $\{0,1,3,4,6,7\}$ in column $C_2,C_3$, rows $\{1,2,4,5,7,8\}$ in column $C_4$, and rows $\{2,0,5,3,8,6\}$ in column $C_5$. One can check that the accessed elements are sufficient to rebuild the two erased columns, and the ratio is $2/3=e/r$. It can be shown that similar rebuilding can be done for any two systematic node erasures.  Therefore, in this example the lower bound is achievable.
\end{xmpl}

Consider an information array of size $p\times k$ and an $(n,k)$ MDS code with $r=n-k$ parity nodes. Each parity node $l\in [0,r-1]$ is constructed from the set of permutations $\{f_i^l\}$ for $i \in [0,k-1]$. Notice that in the general case the number of rows $p$ in the array is not necessarily a power of $r$. We will assume columns $[0,e-1]$ are erased. In an erasure of  $e$ columns, $ep$ elements need rebuilt, hence we need $ep$ equations (zigzags) that contain these elements. In an optimal rebuilding, each parity node contributes $ep/r$ equations by accessing the values of  $ep/r$ of its zigzag elements. Moreover, the union of the zigzag sets that create these zigzag elements, constitute an $e/r$ portion of the elements in the surviving systematic nodes. In other words, assume that we access rows $X$ from the surviving columns $[e,k-1]$, $X \subseteq [0,p-1]$, then $|X|=ep/r$ and
\begin{equation*}
f_j^l(X)=f_i^l(X)
\end{equation*}
for any parity node $l \in [0,r-1]$ and $i,j \in [e,k-1]$. Note that it is equivalent that for any
parity node $l \in [0,r-1]$ and surviving systematic node $j \in [e,k-1]$
\begin{equation*}
f_j^l(X)=f_e^l(X).
\end{equation*}
Let $G^l$ be the subgroup of the symmetric group $S_p$ that is generated by the set of permutations $\{f^{-l}_e\circ f^l_j\}_{j=e}^{k-1}$. It is easy to see that the previous condition is also equivalent to that for any parity $l\in [0,r-1]$ the group $G^l$ \emph{stabilizes} $X$, i.e., for any $f\in G^l,f(X)=X.$

Assuming there is a set $X$ that satisfies this condition, we want to rebuild the $ep$ elements from the chosen $ep$ equations, i.e., the $ep$ equations with the $ep$ variables being solvable. A necessary condition is that each element in the erased column will appear at least once in the chosen zigzag sets (equations). parity $l\in [0,r-1]$ accesses its zigzag elements $f^l_e(X)$, and these zigzag sets contain the elements in rows $(f_i^l)^{-1} f_e^l(X)$  of the erased column $i\in [0,e-1]$. Hence the condition is equivalent to that for any erased column $i\in [0,e-1]$
\begin{equation*}
\cup_{l=0}^{r-1} (f_i^l)^{-1} f_e^l(X) = [0,p-1].
\end{equation*}
These two conditions are necessary for optimal rebuilding ratio. In addition, we need to make sure that the $ep$ equations are linearly independent, which depends on the coefficients in the linear combinations that created the zigzag elements. We summarize:\\
{\bf Sufficient and necessary conditions for optimal rebuilding ratio in $e$ erasures:}
There exists a set $X \subseteq [0,p-1]$ of size $|X|=ep/r$, such that
\begin{enumerate}
	\item For any parity node $l\in [0,e-1]$ the group $G^l$ stabilizes the set $X$, i.e., for any $g\in G^l$
	\begin{equation}
   		g(X)=X,
			\label{eq03171}
		\end{equation}
		where $G^l$ is generated by the set of permutations
		
		$\{f^{-l}_e\circ f^l_j\}_{j=e}^{k-1}.$
  \item For any erased column $i\in [0,e-1]$,
  	\begin{equation}
  		\cup_{l=0}^{r-1} (f_i^l)^{-1} f_e^l(X) = [0,p-1].
  		\label{eq03172}
		\end{equation}
  \item The $ep$ equations (zigzag sets) defined by the set $X$ are linearly independent.
\end{enumerate}

The previous discussion gave the condition for optimal rebuilding ratio in an MDS optimal update code with $e$ erasures in general. Next will interpret these conditions in the special case where the number of rows $p=r^m$, and the permutations are generated by $T=\{v_0,v_1,\dots,v_{k-1}\}$ $\subseteq$ $\mathbb{F}_r^m$ and Construction \ref{cnstr5}, i.e., $f_i^l(x)=x+l v_i$ for any $x \in [0,r^m-1]$.
Note that in the case of $r$ a prime $$G^1=G^2=...=G^{r-1},$$ and in that case we simply denote the group as $G$. The following theorem gives a simple characterization for sets that satisfy condition $1$.

\begin{thm}
Let $X\subseteq \mathbb{F}_r^m$ and $G$ defined above then $G$ stabilizes $X$, if and only if $X$ is a union of cosets of the subspace
\begin{equation}
 \label{eq03242}
Z=\spun\{v_{e+1}-v_e,\dots,v_{k-1}-v_e\}.
\end{equation}
\end{thm}

\begin{IEEEproof}
It is easy to check that any coset of $Z$ is stabilized by $G$, hence if $X$ is a union of cosets it is also a stabilized set. For the other direction let $x,y\in \mathbb{F}^m_r$ be two vectors in the same coset of $Z$, it is enough to show that if $x\in X$ then also $y\in X$. Since $y-x\in Z$ there exist $\alpha_1,...,\alpha_{k-1-e}\in [0,r-1]$ such that
$y-x=\sum_{i=1}^{k-1-e}\alpha_i(v_{e+i}-v_e).$ Since $f(X)=X$ for any $f\in G$ we get that $f(x)\in X$ for any $x\in X$ and  $f\in G$, hence
\begin{align*}y&=x+y-x\\
&=x+\sum_{i=1}^{k-1-e}\alpha_i(v_{e+i}-v_e)\\
&=f_e^{-\alpha_{k-1-e}}f_{k-1}^{\alpha_{k-1-e}}...f_e^{-\alpha_{1}}f_{e+1}^{\alpha_{1}}(x),
\end{align*}
for $f_e^{-\alpha_{k-1-e}}f_{k-1}^{\alpha_{k-1-e}}...f_e^{-\alpha_{1}}f_{e+1}^{\alpha_{1}}\in G$. So $y \in X$ and the result follows.
\end{IEEEproof}
\textbf{Remark:} For any set of vectors $S$ and $v,u\in S$, $$\spun\{S-v\}=\spun\{S-u\}.$$ Here $S-v=\{v_i-v|v_i\in S\}.$ Hence, the subspace $Z$ defined in the previous theorem does not depend on the choice of the vector $v_e$. By the previous theorem we interpret the {\bf necessary and sufficient conditions of an optimal code} as follows:\\
There exists a set $X \subseteq \mathbb{F}_r^m$ of size $|X|=er^{m-1}$, such that
\begin{enumerate}
	\item $X$ is a union of cosets of $$Z=\spun\{v_{e+1}-v_e,\dots,v_{k-1}-v_e\}.$$
	\item For any erased column $i\in [0,e-1]$,
  	\begin{equation}
  	\label{eq03241}
  		\cup_{l=0}^{r-1} (X+l(v_i-v_e))=\mathbb{F}_r^m.	
		\end{equation}
  \item The $er^m$ equations (zigzag sets) defined by the set $X$ are linearly independent.
\end{enumerate}

The following theorem gives a simple equivalent condition for conditions $1,2$.

\begin{thm} \label{thm03311}
There exists a set $X \subseteq \mathbb{F}_r^m$ of size $|X|=er^{m-1}$ such that conditions $1,2$ are satisfied if and only if
\begin{equation} \label{eq03312}
 v_i-v_e \notin Z,
\end{equation}
for any erased column $i\in [0,e-1].$
\end{thm}
\begin{IEEEproof}
Assume conditions $1,2$ are satisfied. If $v_i-v_e\in Z$ for some erased column $i\in [0,e-1]$ then
$X=\cup_{l=0}^{r-1} (X+l(v_i-v_e))=\mathbb{F}_r^m,	$
which is a contradiction to $X\subsetneq \mathbb{F}_r^m.$
On the other hand, If \eqref{eq03312} is true, then $v_i-v_e$ can be viewed as a permutation that acts on the cosets of $Z$. The number of cosets of $Z$ is $r^m/|Z|$ and this permutation (when it is written in cycle notation) contains  $r^{m-1}/|Z|$ cycles, each with length $r$. For each $i\in [0,e-1]$ choose $r^{m-1}/|Z|$ cosets of $Z$, one from each cycle of the permutation $v_i-v_e.$ In total $er^{m-1}/|Z|$ cosets are chosen for the $e$ erased nodes.
Let $X$ be the union of the cosets that were chosen. It is easy to see that $X$ satisfies condition $2$. If $|X|<er^{m-1}$ (Since there might be cosets that were chosen more than once) add arbitrary $(er^{m-1}-|X|)/|Z|$ other cosets of $Z$, and also condition $1$ is satisfied.
\end{IEEEproof}

In general, if \eqref{eq03312} is not satisfied, the code does not have an optimal rebuilding ratio.
However we can define
\begin{equation}
Z=\spun\{v_i-v_e\}_{i\in I},
\label{eq04131}
\end{equation}
where we assume w.l.o.g. $e\in I$ and  $I\subseteq [e,k-1]$ is a maximal subset of surviving nodes that satisfies for any erased node $j\in [0,e-1], v_j-v_e\notin Z.$
Hence from now on we assume that $Z$ is defined by a subset of surviving nodes $I$. This set of surviving nodes will have  an optimal rebuilding ratio (see Corollary \ref{cor0401}), i.e., in the rebuilding of columns $[0,e-1]$, columns $I$ will access a portion of $e/r$ of their elements. The following theorem gives a sufficient condition for the $er^m$ equations defined by the set $X$ to be solvable linear equations.

\begin{thm}\label{thm0413}
Suppose that there exists a subspace $X_0$ that contains $Z$ such that for any erased node $i\in [0,e-1]$
\begin{equation} \label{eq04132}
X_0\oplus \spun\{v_i-v_e\} = \mathbb{F}_r^m,
\end{equation}
then the set $X$ defined as an union of some $e$ cosets of $X_0$ satisfies conditions $1,2$ and $3$ over a field large enough.
\end{thm}
\begin{IEEEproof}
Condition $1$ is trivial.
Note that by \eqref{eq04132}, $l(v_i-v_e) \notin X_0$ for any $l\in [1,r-1]$ and $i\in [0,e-1]$, hence $\{X_0+l(v_i-v_e)\}_{l\in [0,r-1]}$ is the set of cosets of $X_0$. Let $X_j=X_0+j(v_i-v_e)$
be a coset of $X_0$ for some $i \in [0,e-1]$ and suppose $X_j \subset X$. Now let us check condition 2:
\begin{align}
\cup_{l=0}^{r-1}(X+l(v_i-v_e))& \supseteq \cup_{l=0}^{r-1}(X_j+l(v_i-v_e)) \nonumber\\
&= \cup_{l=0}^{r-1}(X_0+j(v_i-v_e)+l(v_i-v_e)) \nonumber \\
&= \cup_{l=0}^{r-1}(X_0+(j+l)(v_i-v_e)) \nonumber \\
& = \cup_{t=0}^{r-1}(X_0+t(v_i-v_e)) \label{05192}\\
&=\mathbb{F}_r^m. \label{05193}
\end{align}
\eqref{05192} holds since $j+l$ is computed mod $r$.
So condition $2$ is satisfied. Next we prove condition $3$.
There are $er^m$ unknowns and $er^m$ equations. Writing the equations in a matrix form we get $AY=b$, where $A$ is an $er^m \times er^m$ matrix. $Y,b$ are vectors of length $er^m$, and $Y=(y_1,...,y_{er^m})^T$ is the unknown vector. The matrix $A=(a_{i,j})$ is defined as $a_{i,j}=x_{i,j}$ if the unknown $y_j$ appears in the $i$-th equation, otherwise $a_{i,j}=0$. Hence we can solve the equations if and only if there is assignment for the indetermediates $\{x_{i,j}\}$ in the matrix $A$ such that $\det(A)\neq 0$.
By \eqref{05193}, accessing rows corresponding to any coset $X_j$ will give us equations where each unknown appears exactly once. Since $X$ is a union of $e$ cosets, each unknown appears $e$ times in the equations. Thus each column in $A$ contains $e$ indeterminates. Moreover, each equation contains one unknown from each erased node, thus any row in $A$ contains $e$ indeterminates.
Then by Hall's Marriage Theorem \cite{Hall} we conclude that there exists a permutation $f$ on the integers $[1,er^m]$ such that $$\prod_{i=1}^{er^m}a_{i,f(i)}\neq 0.$$ Hence the polynomial $\det(A)$ when viewed as a symbolic polynomial, is not the zero polynomial, i.e., $$\det(A)=\sum_{f\in S_{er^m}}\sgn(f)\prod_{i=1}^{er^m}a_{i,f(i)}\neq 0.$$ By Theorem \ref{polynomial-method} we conclude that there is an assignment from a field large enough for the indeterminates such that $\det(A)\neq 0$, and the equations are solvable. Note that this proof is for a specific set of erased nodes. However if \eqref{eq04132} is satisfied for any set of $e$ erasures, multiplication of all the nonzero polynomials $\det(A)$ derived for any set of erased nodes is again a nonzero polynomial and by the same argument there is an assignment over a field large enough such that any of the  matrices $A$ is invertible, and the result follows.
\end{IEEEproof}

In order to use Theorem \ref{thm0413}, we need to find a subspace $X_0$ as in \eqref{eq04132}. The following theorem shows that such a subspace always exists, moreover it gives an explicit construction of it.

\begin{thm}
\label{thm0414}
Suppose $1 \le e<r$ erasures occur.
Let $Z$ be defined by \eqref{eq04131} and $v_i-v_e \notin Z$ for any erased node $i \in [0,e-1]$. Then there exists $u\perp Z$ such that for any $i \in [0,e-1]$,
\begin{equation}\label{eq04144}
u \cdot (v_i-v_e) \neq 0.
\end{equation}
Moreover the orthogonal subspace $X_0=(u)^\perp$ satisfies \eqref{eq04132}.
\end{thm}

\begin{IEEEproof}
First we will show that such vector $u$ exists. Let $u_1,...u_t$ be a basis for $(Z)^\perp$ the orthogonal subspace of $Z.$
Any vector $u$ in $(Z)^\perp$ can be written as $u = \sum_{j=1}^t x_j u_j$ for some $x_j$'s.
We claim that for any $i \in [0,e-1]$ there exists $j$ such that $u_j \cdot(v_i-v_e) \neq 0$. Because otherwise, $(Z)^{\perp}=\spun\{u_1,\dots,u_t\} \perp v_i-v_e$, which means $v_i-v_e \in Z$ and reaches a contradiction. Thus the number of solutions for the linear equation
$$\sum_{j=1}^{t}x_ju_j \cdot (v_i-v_e)=0$$
is $r^{t-1}$, which equals the number of $u$ such that $u\cdot (v_i-v_e)=0$.
 Hence by the union bound there are at most $er^{t-1}$ vectors $u$ in $(Z)^\perp$ such that $u\cdot (v_i-v_e)=0$ for some erased node $i\in [0,e-1]$. Since $|(Z)^\perp|=r^t>er^{t-1}$ there exists $u$ in $(Z)^\perp$ such that for any erased node $i\in [0,e-1]$, $$u\cdot (v_i-v_e)\neq 0.$$ Define $X_0=(u)^\perp$, and note that for any erased node $i\in [0,e-1],v_i-v_e\notin X_0$, since $u\cdot (v_i-v_e)\neq 0$ and $X_0$ is the orthogonal subspace of $u$. Moreover, since $X_0$ is a hyperplane we conclude that
$$X_0\oplus \spun\{v_i-v_e\}=\mathbb{F}_r^m,$$
and the result follows.
\end{IEEEproof}

Theorems \ref{thm0413} and \ref{thm0414} give us {\bf an algorithm to rebuild multiple erasures:}
\begin{enumerate}
\item
Find $Z$ by \eqref{eq04131} satisfying \eqref{eq03312}.
\item
Find $u \perp Z$ satisfying \eqref{eq04144}. Define $X_0=(u)^{\perp}$ and $X$ as a union of $e$ cosets of $X_0$.
\item
Access rows $f_e^l(X)$ in parity $l \in [0,r-1]$ and all the corresponding information elements.
\end{enumerate}
We know that under a proper selection of coefficients the rebuilding is possible.

In the following we give two examples of rebuilding using this algorithm.
The first example shows an optimal rebuilding for any set of $e$ node erasures. As mentioned above, the optimal rebuilding is achieved since \eqref{eq03312} is satisfied, i.e., $I=[e,k-1]$.
\begin{xmpl} \label{xmpl0331}
Let $T=\{v_0,v_1,\dots,v_m\}$ be a set of vectors that contains an orthonormal basis of $\mathbb{F}_r^m$ together with the zero vector. Suppose columns $[0,e-1]$ are erased. Note that in that case $I=[e,m]$ and $Z$ is defined as in \eqref{eq04131}. Define $$u=\sum_{j=e}^mv_j,$$
and $X_0=(u)^\perp$.
When $m=r$ and $e=r-1$, modify $u$ to be $$u=\sum_{i=1}^{m}v_i.$$
It is easy to check that $u\perp Z$ and for any erased column $i\in [0,e-1],u\cdot (v_i-v_e)=-1$. Therefore by Theorems \ref{thm0413} and \ref{thm0414} a set $X$ defined as a union of an arbitrary $e$ cosets of $X_0$ satisfies conditions $1,2$ and $3$, and optimal rebuilding is achieved.
\end{xmpl}

In the example of Figure \ref{fig5}, we know that the vectors generating the permutations are the standard basis (and thus are orthonormal basis) and the zero vector. When columns $C_0,C_1$ are erased, $u=e_2$ and $X_0=(u)^{\perp}=\spun\{e_1\}=\{0,3,6\}$. Take $X$ as the union of $X_0$ and its coset $\{1,4,7\}$, which is the same as Example \ref{xmpl0503}. One can check that each erased element appears exactly 3 times in the equations and the equations are solvable in $\mathbb{F}_7$. Similarly, the equations are solvable for other $2$ systematic erasures.

Before we proceed to the next example, we give an upper bound for the rebuilding ratio using Theorem \ref{thm0413} and a set of nodes $I.$
\begin{cor} \label{cor0401}
Theorem \ref{thm0413} requires rebuilding ratio at most
$$\frac{e}{r}+\frac{(r-e)(k-|I|-e)}{r(k+r-e)}$$
\end{cor}
\begin{IEEEproof}
By Theorem \ref{thm0413}, the fraction of accessed elements in columns $I$ and the parity columns is $e/r$ of each column. Moreover, the accessed elements in the rest columns are at most an entire column. Therefore, the ratio is at most
$$\frac{\frac{e}{r}(|I|+r)+(k-|I|-e)}{k+r-e}
= \frac{e}{r}+\frac{(r-e)(k-|I|-e)}{r(k+r-e)}$$
and the result follows.
\end{IEEEproof}

Note that as expected when $|I|=k-e$ the rebuilding ratio is optimal, i.e. $e/r$.
In the following example the code has $O(m^2)$ columns. The set $I$ does not contain all the surviving systematic nodes, hence the rebuilding is not optimal but is at most $\frac{1}{2}+O(\frac{1}{m})$.

\begin{xmpl}
 \label{xmpl0401}
Suppose $2|m$. Let $T=\{v=(v_1,\dots,v_m):\|v\|_1=2,v_i=1,v_j=1, \text{ for some } i \in [1,m/2], j \in [m/2+1,m]\}\subset \mathbb{F}_2^m$ be the set of vectors generating the code with $r=2$ parities, hence the number of systematic nodes is $|T|=k=m^2/4$. Suppose column $w=(w_1,\dots,w_m)$, $w_1=w_{m/2+1}=1$ is erased. Define the set $I=\{v\in T: v_1=0\},$ and
$$Z=\spun\{v_i-v_e|i\in I\}$$
for some $e \in I$. Thus $|I|=m(m-2)/4$.
It can be seen that $Z$ defined by the set $I$ satisfies \eqref{eq03312}, i.e., $w-v_e\notin Z$  since the first coordinate of a vector in $Z$ is always $0$, as oppose to $1$ for the vector $w-v_e$. Define $u=(0,1,...,1)$ and $X_0=(u)^\perp.$ It is easy to check that $u\perp Z$ and   $u\cdot (w-v_e)= 1\neq 0.$ Hence, the conditions in Theorem \ref{thm0414} are satisfied and rebuilding can be done using $X_0$.
Moreover by Corollary \ref{cor0401} the rebuilding ratio is at most
$$\frac{1}{2}+\frac{1}{2}\frac{(m/2)-1}{(m^2/4)+1} \approx \frac{1}{2}+\frac{1}{m},$$
which is a little better than Theorem \ref{thm:k/3} in the constants. Note that by similar coefficients assignment of Construction \ref{k/3}, we can use a field of size $5$ or $8$ to assure the code will be an MDS code.
\end{xmpl}

\subsection{Minimum Number of Erasures with Optimal Rebuilding}
Next we want to point out a surprising phenomena.
We say that a set of vectors $S$ satisfies \emph{property} $e$ for $e\geq 1$ if for any subset $A\subseteq S $ of size $e$ and any $u\in A$,
$$u-v\notin \spun\{w-v:w\in S\backslash A\},$$
where $v\in S\backslash A$. Recall that by Theorem \ref{thm03311} any set of vectors that generates a code $\cC$ and can rebuild optimally any $e$ erasures, satisfies  property $e$. The following theorem shows that this property is monotonic, i.e., if $S$ satisfies property $e$ then it also satisfies property $a$ for any $e\leq a \leq |S|.$

\begin{thm}
Let $S$ be a set of vectors that satisfies property $e$, then it also satisfies  property $a$, for any $e \leq a \leq |S|$.
\end{thm}

\begin{IEEEproof}
Let $A\subseteq S,|A|=e+1$ and assume to the contrary that $u-v\in \spun\{w-v:w\in S\backslash A\}$ for some $u\in A$ and $v\in S\backslash A$. $|A|\geq 2$ hence there exists $x\in A\backslash u$. It is easy to verify that $u-v\in \spun\{w-v:w\in S\backslash A^*\}$, where $A^*=A\backslash x$ and $|A^*|=e$ which contradicts the property $e$ for the set $S$.
\end{IEEEproof}
Hence, from the previous theorem we conclude that a code $\cC$ that can rebuild optimally $e$ erasures, is able to rebuild optimally any number of erasures greater than $e$ as well. However, as pointed out  already there are codes with $r$ parities that can not rebuild optimally from some $e<r$ erasures. Therefore, one might expect to find a code $\cC$ with parameter $e^* \ge 1$ such that it can rebuild optimally $e$ erasures \emph{only} when $e^*\leq e\leq r$. For example, for $r=3,m=2$ let $\cC$ be the code constructed by the vectors $\{0,e_1,e_2,e_1+e_2\}$. We know that any code with more than $3$ systematic nodes can not rebuild one erasure optimally, since the size of a family of orthogonal permutations over the integers $[0,3^2-1]$ is at most $3$. However, one can check that for any two erased columns, the conditions in Theorem \ref{thm0413} are satisfied hence the code can rebuild optimally for any $e=2$ erasures and we conclude that  $e^*=2$ for this code.

The phenomena that some codes has a threshold parameter $e^*$, such that \emph{only} when the number of erasures $e$ is at least as the threshold $e^*$ then the code can rebuild optimally, is a bit counter intuitive and surprising. This phenomena gives rise to another question. We know that for a code constructed with vectors from $\mathbb{F}_r^m$, the maximum number of systematic columns for optimal rebuilding of $e=1$ erasures is $m+1$ (Theorem \ref{thm0519}). Can the number of systematic columns in a code with an optimal rebuilding of $e>1$ erasures be increased? The previous example shows a code with $4$ systematic columns can rebuild optimally any $e=2$ erasures.
But Theorem \ref{thm0519} shows that when $r=3,m=2$, optimal rebuilding for $1$ erasure implies no more than $3$ systematic columns.
 Hence the number of systematic columns is increased by at least $1$ compared to codes with $9$ rows and optimal rebuilding of $1$ erasure. The following theorem gives an upper bound for the maximum systematic columns in a code that rebuilds optimally any $e$ erasures.

\begin{thm}
Let $\cC$ be a code constructed by Construction \ref{cnstr5} and vectors from $\mathbb{F}_r^m$. If $\cC$ can rebuild optimally any $e$ erasures, for some $1\le e<r$, then  the number of systematic columns $k$ in the code satisfies
$$k \le m+e.$$
\end{thm}
\begin{IEEEproof}
Consider a code with length $k$ and generated by vectors $v_0,v_1,\dots,v_{k-1}.$ If these vectors are linearly independent then $k\leq m$ and we are done. Otherwise they are dependent.
Suppose $e$ columns are erased,  $1 \le e<r$. Let $v_{e}$ be a surviving column. Consider a new set a of vectors: $T=\{v_i-v_{e}: i \in [0,k-1],i\neq e\}.$
We know that the code can rebuild optimally only if \eqref{eq03312} is satisfied for all possible $e$ erasures. Thus for any $i \neq e$, $i \in [0,k-1],$ if column $i$ is erased and column $e$ is not, we have
$v_i-v_e \notin Z$ and thus $v_i-v_e \neq 0$.
So every vector in $T$ is nonzero.
Let $s$ be the minimum number of dependent vectors in $T$, that is, the minimum number of vectors in $T$ such that they are dependent.
For nonzero vectors, we have $s \ge 2$.
Say $\{v_{e+1}-v_{e},v_{e+2}-v_e,\dots,v_{e+s}-v_e\}$ is a minimum dependent set of vector. Since any $m+1$ vectors are dependent in $\mathbb{F}_r^m$,
$$s \le m+1.$$
We are going to show $k-e \le s-1$.
Suppose to the contrary that the number of remaining columns satisfies $k-e \ge s$ and $e$ erasures occur. When column $v_{e+s}$ is erased and the $s$ columns $\{v_e,v_{e+1},\dots,v_{e+s-1}\}$ are not, we should be able to rebuild optimally.
However since we chose a dependent set of vectors, $v_{e+s}-v_e$ is a linear combination of $\{v_{e+1}-v_{e},v_{e+2}-v_e,\dots,v_{e+s-1}-v_e\}$, whose span is contained in $Z$ in \eqref{eq03312}. Hence \eqref{eq03312} is violated and we reach a contradiction. Therefore,
$$k-e \le s-1 \le m.$$
\end{IEEEproof}

Notice that this upper bound is tight. For $e=1$ we already gave codes with optimal rebuilding of $1$ erasure and $k=m+1$ systematic columns. Moreover, for $e=2$ the code already presented in this section and constructed by the vectors $0,e_1,e_2,e_1+e_2$, reaches the upper bound with $k=4$ systematic columns.

\subsection{Generalized Rebuilding Algorithms}

The rebuilding algorithms presented in Constructions \ref{cnstr1},\ref{cnstr5} and Theorem \ref{thm0413} all use a specific subspace and its cosets in the rebuilding process. This method of rebuilding can be generalized by using an arbitrary subspace as explained below.

Let $T=\{v_0,\dots,v_{k-1}\}$ be a set of vectors generating the code in Construction \ref{cnstr5} with $r^m$ rows and $r$ parities. Suppose $e$ columns $[0,e-1]$ are erased. Let $Z$ be a proper subspace of $\mathbb{F}_r^m$. In order to rebuild the erased nodes, in each parity column $l \in [0,r-1]$,  access the zigzag elements $z_i^l$ for $i\in X_l$, and $X_l$ is a union of cosets of $Z$.
In each surviving node,  access all the elements that are in the zigzag sets $X_l$ of parity $l$. More specifically,  access element $a_{i,j}$ in the surviving column $j\in [e,k-1]$ if $i+l v_j  \in X_l$.
Hence, in the surviving column $j$ and parity $l$, we access elements in rows $X_l-lv_j$. In order to make the rebuilding successful we impose the following conditions on the sets $X_0,...,X_l$. Since the number of equations needed is at least as the number of erased elements, we require
\begin{equation}\label{eq04033}
\sum_{l=0}^{r-1} |X_l|= er^{m}.
\end{equation}
Moreover, we want the equations to be solvable, hence for any erased column $i\in [0,e-1]$,
\begin{equation}\label{eq04032}
 \cup_{l=0}^{r-1}X_l-lv_i=[0,r^m-1] \text{ multiplicity } e,
\end{equation}
which means if the union is viewed as a multi-set, then each element in $[0,r^m-1]$ appears exactly $e$ times. This condition makes sure that the equations are solvable by Hall's theorem (see Theorem \ref{thm0413}). Under these conditions we would like to minimize the ratio, i.e., the number of accesses which is,
\begin{equation}\label{eq04031}
\min_{X_0,\dots,X_{r-1}} \sum_{j=e}^{k-1} |\cup_{l=0}^{r-1} (X_l-lv_j)|.
\end{equation}
In summary, for the {\bfseries generalized rebuilding algorithm} one first chooses a subspace $Z$, and then solves the minimization problem in \eqref{eq04031} subject to \eqref{eq04033} and \eqref{eq04032}.

The following example interprets the minimization problem for a specific case.
\begin{xmpl}
Let $r=2,e=1$, i.e., two parities and one erasure, then equations \eqref{eq04033},\eqref{eq04032} becomes
\begin{equation*}
|X_0|+|X_1|=2^m,\\
X_0\cup X_1+v_0=[0,2^m-1].
\end{equation*}
Therefore $X_1+v_0=\overline{X_0}.$ The objective function in \eqref{eq04031} becomes,
\begin{eqnarray*}
\min_{X_0,X_1} \sum_{j=1}^{k-1} |X_0 \cup X_1+v_j| = \min_{X_0} \sum_{j=1}^{k-1} |X_0 \cup(\overline{X_0}+v_0+v_j)|.
\end{eqnarray*}
Each $v_0+v_j$ defines a permutation $f_{v_0+v_j}$ on the cosets of $Z$ by $f_{v_0+v_j}(A)=A+v_0+v_j$ for a coset $A$ of $Z$. If $v_0+v_j\in Z$ then $f_{v_0+v_j}$  is the identity permutation and $|X_0 \cup(\overline{X_0}+v_0+v_j)|=2^m$, regardless of the choice of $X_0$. However, if $v_0+v_j\notin Z$, then $f_{v_0+v_j}$ is of order $2$, i.e., it's composed of disjoint cycles of length $2$. Note that if $f_{v_0+v_j}$ maps $A$ to $B$ and only one of the cosets  $A,B$ is contained in $X_0$, say $A$, then only $A$ is contained in $X_0 \cup(\overline{X_0}+v_0+v_j)$. On the other hand, if both $A,B\in X_0 \text{ or } A,B\notin X_0$ then, $$A,B\subseteq X_0 \cup(\overline{X_0}+v_0+v_j).$$ In other words, $(A,B)$ is a cycle in $f_{v_0+v_j}$ which is totally contained in $X_0\text{ or in } \overline{X_0}$.  Define $N_j^X$ as the number of cycles $(A,B)$ in the permutation $f_{v_0+v_j}$ that are totally contained in $X$ or in $\overline{X}$, where $X$ is a union of some cosets of
$Z$. It is easy to see that the minimization problem is equivalent to minimizing
\begin{equation} \label{0512}
\min_{X} \sum_{j=1}^{k-1}N_j^X.
\end{equation}
In other words, we want to find a set $X$ which is a union of cosets of $Z$, such that the number of totally contained or totally not contained cycles in the permutations defined by $v_j+v_0$, $j \in [1,k-1]$ is minimized.
\end{xmpl}
From the above example, we can see that given a non-optimal code with two parities and one erasure, finding the solution in \eqref{0512} requires minimizing for the sum of these $k-1$ permutations, which is an interesting combinatorial problem. Moreover, by choosing a different subspace $Z$  we might be able to get a better rebuilding algorithm than that in Construction \ref{cnstr1} or Theorem \ref{thm0413}.

%%%%%%%%%%%%%%%%%%%%%%%%%%%%%%%%%%%%%%%%%%%%%%%%%%%%%%%%%%%%

\section{Concluding Remarks}
\label{summary}

%%%%%%%%%%%%%%%%%%%%%%%%%%%%%%%%%%%%%%%%%%%%%%%%%%%%%%%%%%%
In this paper, we described explicit constructions of the first known systematic $(n,k)$ MDS array codes with $n-k$ equal to some constant, such that the amount of information needed to rebuild an erased column equals to $1/(n-k)$, matching the information-theoretic lower bound.
While the codes are new and interesting from a theoretical perspective, they also provide an exciting practical solution, specifically, when $n-k=2$, our zigzag codes are the best known alternative to RAID-6 schemes. RAID-6 is the most prominent scheme in storage systems for combating disk failures\cite{Shuki-evenodd}-\cite{star-code}. Our new zigzag codes provide a RAID-6 scheme that has optimal update (important for write efficiency), small finite field size (important for computational efficiency) and optimal access of information for rebuilding - cutting the current rebuilding time by a factor of two.

We note that one can add redundancy for the sake of lowering the rebuilding ratio. For instance, one can use three parity nodes instead of two. The idea is that the third parity is not used for protecting data from erasures, since in practice, three concurrent failures are unlikely. However, with three parity nodes, we are able to rebuild a single failed node by accessing only $1/3$ of the remaining information (instead of $1/2$). An open problem is to construct codes that can be extended in a simple way, namely, codes with three parity nodes such that the first two nodes ensure a rebuilding ratio of $1/2$ and the third node further lowers the ratio to $1/3$. Hence, we can first construct an array with two parity nodes and when needed, extend the array by adding an additional parity node to obtain additional improvement in the rebuilding ratio.

Another future research direction is to consider the ratio of read accesses in the case of a write (update) operation.  For example, in an array code with two parity nodes, in order to update a single information element, one needs to read at least three elements and write three elements, because we need to know the values of the old information and old parities and compute the new parity elements (by subtracting the old information from the parity and adding the new information). However, an interesting observation, in our optimal code construction with two parity nodes, is if we update all the information in the first column and the rows in the first half of the array (see Figure \ref{fig2}), we do not need to read for computing the new parities, because we know the values of all the information elements needed for computing the parities. These information elements take about half the size of the entire array. So in a storage system we can cache the information to be written until most of these elements needs update (we could arrange the information in a way that these elements are often updated at the same time), hence, the ratio between the number of read operations and the number of new information elements is relatively very small. Clearly, we can use a similar approach for any other systematic column. In general, given $r$ parity nodes, we can avoid redundant read operations if we update about $1/r$ of the array.

%%%%%%%%%%%%%%%%%%%%%%%%%%%%%%%%%%%%%%%%%%%%%%%%%%%%%%%%%%%%%%%%%%%
%
%
\appendices
%
%
%%%%%%%%%%%%%%%%%%%%%%%%%%%%%%%%%%%%%%%%%%%%%%%%%%%%%%%%%%%%%%%%%%%
\section{Proof of Theorem \ref{zigzag-sets}}
\label{app1}
In order to prove Theorem \ref{zigzag-sets},
we use the well known Combinatorial Nullstellensatz by Alon \cite{Alon-polynomial-method}:
\begin{thm} (Combinatorial Nullstellensatz) \cite[Th 1.2]{Alon-polynomial-method}
\label{polynomial-method}
Let $\mathbb{F}$ be an arbitrary field, and let $f=f(x_1,...,x_q)$ be a polynomial in $\mathbb{F}[x_1,...,x_q]$. Suppose the degree of $f$ is $\deg(f)=\sum_{i=1}^q t_i$, where each $t_i$ is a nonnegative integer, and suppose the coefficient of $\prod_{i=1}^q x_i^{t_i}$ in $f$ is nonzero. Then, if $S_1,...,S_n$ are subsets of $\mathbb{F}$ with $|S_i| > t_i$, there are
$s_1\in S_1,s_2\in S_2,...,s_q\in S_q$ so that $$f(s_1,...,s_q)\neq 0.$$
\end{thm}

\begin{IEEEproof}[Proof of Theorem \ref{zigzag-sets}]
Assume the information of $A$ is given in a column vector $W$ of length $pk$, where column $i\in [0,k-1]$ of $A$ is in the row set $[(ip,(i+1)p-1]$ of $W$. Each systematic node $i$, $i \in [0,k-1]$, can be represented as $Q_iW$ where $Q_i=[0_{p\times pi},I_{p\times p},0_{p\times p(k-i-1)}]$. Moreover define $Q_{k}=[I_{p\times p}, I_{p\times p},...,I_{p\times p}],Q_{k+1}=[x_0P_0,x_1P_1,...,x_{k-1}P_{k-1}]$ where the $P_i$'s are permutation matrices (not necessarily distinct) of size $p\times p$, and the $x_i$'s are variables, such that $C_{k}=Q_{k}W,C_{k+1}=Q_{k+1}W$. The permutation matrix  $P_i=(p^{(i)}_{l,m})$ is defined as $p^{(i)}_{l,m}=1$ if and only if $a_{m,i}\in Z_l$. In order to show that there exists such MDS code, it is sufficient to show that there is an assignment for the intermediates $\{x_i\}$ in the field $\mathbb{F}$, such that for any set of integers $\{s_1,s_2,...,s_k\}\subseteq [0,k+1]$ the matrix $Q=[Q_{s_1}^T,Q_{s_1}^T,...,Q_{s_k}^T ]$ is of full rank. It is easy to see that if the parity column $C_{k+1}$ is erased i.e., $k+1 \notin \{s_1,s_2,...,s_k\}$ then $Q$ is of full rank. If $k\notin \{s_1,s_2,...,s_k\}\text{ and } k+1\in \{s_1,s_2,...,s_q\}$ then $Q$ is of full rank if none of the $x_i$'s equals to zero. The last case is when both $k,k+1\in \{s_1,s_2,...,s_k\}$, i.e.,  there are $0\leq i<j\leq k-1$ such that $i,j\notin \{s_1,s_2,...,s_k\}.$ It is easy to see that in that case $Q$ is of full rank if and only if the submatrix
\[B_{i,j}= \left( \begin{array}{cc}
x_iP_i &x_jP_j  \\
I_{p\times p} & I_{p\times p}  \end{array} \right)\]
is of full rank. This is equivalent to $ \det(B_{i,j})\neq 0$. Note that $\deg( \det(B_{i,j}))=p$ and the coefficient of $x_i^p$ is $\det(P_i)~\in \{1,-1\}$.
Define the polynomial $$T=T(x_0,x_1,...,x_{k-1})=\prod_{0\leq i<j\leq k-1}\det(B_{i,j}),$$ and the result follows if there are elements $a_0,a_1,..,a_{k-1}\in \mathbb{F}$ such that $T(a_0,a_1,...,a_{k-1})\neq 0.$ $T$ is of degree $p\binom{k}{2}$ and the coefficient of
$\prod_{i=0}^{k-1} x_i^{p(k-1-i)}$ is $\prod_{i=0}^{k-1} \det(P_i)^{k-1-i}\neq 0$. Set for any $i, S_i=\mathbb{F} \backslash 0$ in Theorem \ref{polynomial-method}, and the result follows.
\end{IEEEproof}

\section{Proof of Theorem \ref{monotone function}}
\label{app2}
In this section we prove Theorem \ref{monotone function}. We will need some definitions and theorems first.

For the rebuilding of node $i$ by row and zigzag sets $\mathbf{S}=\{S_0,\dots,S_{p-1}\}$, define the number of intersections by
$$I(i|\mathbf{S})=\sum_{S\in \mathbf{S}}|S|-|\cup_{S\in \mathbf{S}}S|=pk-|\cup_{S\in \mathbf{S}}S|.$$ Moreover define the number of total intersections in an MDS array code $\mathcal{C}$ as $$I(\mathcal{C})=\sum_{i=0}^{k-1}\max_{\mathbf{S}\text{ rebuilds } i}I(i|\mathbf{S}).$$ Now define $h(k)$ to be the maximal possible intersections over all $(k+2,k)$ MDS array codes, i.e., $$h(k)=\max_{\mathcal{C}}I(\mathcal{C}).$$
For example, in Figure \ref{fig:shapes} the rebuilding set for column $1$ is $\mathbf{S}=\{R_0,R_1,Z_0,Z_1\}$, the size in equation (\ref{eq:77}) is $8$, and $I(1|\mathbf{S})=4$.

The following theorem gives a recursive bound for the maximal number of intersections.

\begin{thm}
\label{th:inequality}
Let $q\leq k\leq p$ then $h(k)\leq \frac{k(k-1)h(q)}{q(q-1)}$.
\end{thm}
\begin{IEEEproof}
Let $A$ be an information array of size $p \times k$. Construct a MDS array code $\cC$ by the row sets and the zigzag sets that reaches the maximum possible number of intersections, and suppose $\mathbf{S}^i$ achieves the maximal number of intersections for rebuilding column $i$, $i \in [0,k-1]$. Namely the zigzag sets $Z$ of the code $\cC$ and the rebuilding sets $\mathbf{S}^i$ satisfy that,
$$h(k)=I(\cC)=\sum_{i=0}^{k-1}\max_{\mathbf{S}\text{ rebuilds } i}I(i|\mathbf{S})=\sum_{i=0}^{k-1} I(i|\mathbf{S}^i).$$
For a subset of columns $T\subseteq [0,k-1]$ and a rebuilding set $\mathbf{S}^i=\{S_0,...,S_{p-1}\}$ we denote the restriction of $\mathbf{S}^i$ to $T$ by $\mathbf{S}^i_T=\{S_{0,T},...,S_{p-1,T}\}$, where
$S_{l,T}=\{a\in S_l:a\text{ is in columns } T\}$. Denote by
$$I(j,\mathbf{S}^i)=\sum_{l=0}^{p-1}|S_l\cap j|-|(\cup_{l=0}^{p-1}S_l)\cap j|$$
the number of intersections in column $j$ while rebuilding column $i$ by $\mathbf{S}^i$.
It is easy to see that
$$I(i|\mathbf{S}^i)=\sum_{j:j \neq i}I(j,\mathbf{S}^i)$$
and thus
$$h(k)=\sum_{i,j:j \neq i}I(j,\mathbf{S}^i).$$
Note also that if $i \neq j$ and $i,j \in T$, then
\begin{equation}
I(j,\mathbf{S}^i)=I(j,\mathbf{S}^i_T).
\label{eq:1.1}
\end{equation}
Hence
\begin{align}
\binom{k-2}{q-2}h(k)&=\binom{k-2}{q-2}\sum_{\substack{i,j:\\j\neq i}}I(j,\mathbf{S}^i)\nonumber\\
&=\sum_{\substack{i,j:\\j\neq i}}\quad\sum_{\substack{T\subseteq [0,k-1]:\\i,j\in T,|T|=q}}I(j,\mathbf{S}^i)\nonumber\\
&=\sum_{\substack{i,j:\\j\neq i}}\quad\sum_{\substack{T\subseteq [0,k-1]:\\i,j\in T,|T|=q}}I(j,\mathbf{S}^i_T)\nonumber\\
&=\sum_{\substack{T\subseteq [0,k-1]:\\|T|=q}}\quad\sum_{\substack{i,j\in T:\\i\neq j}}I(j,\mathbf{S}^i_T)\nonumber\\
&\leq\sum_{\substack{T\subseteq [0,k-1]:\\|T|=q}}h(q) \label{eq0429}\\
&=\binom{k}{q}h(q).\nonumber
\label{eq:34}
\end{align}
Inequality \eqref{eq0429} holds because the code restricted in columns $T$ is a $(q+2,q)$ MDS and optimal-update code, and $h(q)$ is the maximal intersections among such codes.
Hence,
$$h(k)\leq \frac{\binom{k}{q}h(q)}{\binom{k-2}{q-2}}=\frac{k(k-1)h(q)}{q(q-1)},$$
and the result follows.
\end{IEEEproof}

For a $(k+2,k)$ MDS code $\cC$ with $p$ rows the \emph{rebuilding ratio} $R(\mathcal{C})$ can be written as
$$R(\mathcal{C}) = \frac{k(p(k-1)-I(\mathcal{C})+p)}{p(k+1)k}=1-\frac{I(\mathcal{C})+pk}{p(k+1)k}.$$
Notice that in the two parity nodes, we access $p$ elements because each erased element must be rebuilt either by row or by zigzag. Thus we have the term $p$ in the above equation.
And the \emph{ratio function} for all $(k+2,k)$ MDS codes with $p$ rows is $$R(k)=\min_{\cC}{R(\cC)}=1-\frac{h(k)+pk}{p(k+1)k}.$$

\begin{IEEEproof}[Proof of Theorem \ref{monotone function}]
Consider a $(k+2,k)$ code with $p$ rows and assume a systematic node is erased.
In order to rebuild it, $p$ row and zigzag sets are accessed.
Let $x\text{ and } p-x$ be the number of elements that are accessed from the first and the second parity respectively. W.l.o.g we can assume that $x\geq \frac{p}{2}$, otherwise $p-x$ would satisfy it. Each element of these $x$ sets is a sum of a set of size $k$. Thus in order to rebuild the node, we need to access at least $x(k-1)\geq \frac{p(k-1)}{2}$ elements in the $k-1$ surviving systematic nodes, which is at least half of the size of these nodes. So the number of intersections is no more than $\frac{pk(k-1)}{2}$.
Thus
\begin{equation} \label{half}
h(k) \le \frac{pk(k-1)}{2}.
\end{equation}
and the ratio function satisfies
$$R(k)=1-\frac{h(k)+pk}{pk(k+1)}\geq 1-\frac{\frac{pk(k-1)}{2}+pk}{pk(k+1)}=\frac{1}{2}.$$
So the rebuilding ratio is no less than $1/2$.

From Theorem \ref{th:inequality} we get,
\begin{equation}
h(k+1)\leq \frac{(k+1)kh(k)}{k(k-1)}=\frac{(k+1)h(k)}{k-1}.
\label{eq:2.34}
\end{equation}
Hence,
\begin{align}
R(k+1)&=1-\frac{h(k+1)}{p(k+1)(k+2)}-\frac{1}{k+2} \nonumber\\
&\geq  1-\frac{h(k)}{p(k-1)(k+2)}-\frac{1}{k+2} \nonumber\\
& =  1 - \frac{h(k)+p(k-1)}{p(k-1)(k+2)} \nonumber\\
& \ge  1- \frac{h(k)+pk}{pk(k+1)} \label{eq05011}\\
& =  R(k), \nonumber
\end{align}
where \eqref{eq05011} follows from \eqref{half}.
Thus the ratio function is nondecreasing.
\end{IEEEproof}

The lower bound of $1/2$ in the theorem can be also derived from the repair bandwidth \eqref{eq:tradeoff}.

\end{document}